\documentclass{ws-ijfcs}

\usepackage{amsmath}
\usepackage{amssymb}
\usepackage[linesnumbered,vlined,ruled,algo2e]{algorithm2e}

 \newtheorem{case}{Case}
 \newtheorem{subcase}{Subcase}
\newtheorem{observation}{Observation}
\usepackage{chngcntr}

\usepackage{enumerate}
\usepackage{url}
\begin{document}

\markboth{Bhagat et al.}
{$$Optimal Gathering over Weber Meeting Nodes in Infinite Grid \/$$}

%
\catchline{}{}{}{}{}
%

\title{Optimal Gathering over Weber Meeting Nodes in Infinite Grid}



\author{Subhash Bhagat}
\address{National Institute of Science Education and Research\\ HBNI, Bhubaneswar, India\\
\email{subhash.bhagat.math@gmail.com}}

\author{Abhinav Chakraborty}
\address{Advanced Computing and Microelectronics Unit\\Indian Statistical Institute, Kolkata, India\\
abhinav.chakraborty06@gmail.com
}

\author{Bibhuti Das}
\address{Advanced Computing and Microelectronics Unit\\Indian Statistical Institute, Kolkata, India\\
dasbibhuti905@gmail.com
}

\author{Krishnendu Mukhopadhyaya}
\address{Advanced Computing and Microelectronics Unit\\Indian Statistical Institute, Kolkata, India\\
krishnendu.mukhopadhyaya@gmail.com
}

\maketitle
		
		\begin{abstract}
		The \textit{gathering over meeting nodes} problem requires the robots to gather at one of the pre-defined \textit{meeting nodes}. This paper investigates the problem with respect to the objective function that minimizes the total number of moves made by all the robots. In other words, the sum of the distances traveled by all the robots is minimized while accomplishing the \textit{gathering} task. The robots are deployed on the nodes of an anonymous two-dimensional infinite grid which has a subset of nodes marked as \textit{meeting nodes}. The robots do not agree on a global coordinate system and operate under an asynchronous scheduler. A deterministic distributed algorithm has been proposed to solve the problem for all those solvable configurations, and the initial configurations for which the problem is unsolvable have been characterized. The proposed \textit{gathering} algorithm is optimal with respect to the total number of moves performed by all the robots in order to finalize the \textit{gathering}.
			\end{abstract}
		

\section{Introduction}
	One recent trend in robotics is using a group of small, inexpensive and mass-produced robots to perform complex tasks. The main focus of theoretical research in swarm robotics is to identify minimal sets of capabilities necessary to solve a particular problem. The {\it gathering} problem asks the mobile entities, which are initially situated at distinct locations, to gather at a common location and remain there within a finite amount of time. In this paper, the robots are deployed on the nodes of an anonymous grid graph. 
	
	\noindent Robots are assumed to be anonymous (no unique identifiers), autonomous (without central control), homogeneous (execute the same deterministic algorithm) and oblivious (no memory of past observations). They do not have any explicit means of communication, i.e., they cannot send any messages to other robots. They do not have any agreement on a global coordinate system or chirality. Each robot has its own local coordinate system with the robot's current position as the origin. They are equipped with sensor capabilities in order to observe the positions of the other robots. No local memory is available on the nodes of the grid graph. The robots have unlimited visibility, i.e., they can perceive the entire graph. 
	
	\noindent Robots operate in Look-Compute-Move (LCM) cycles. In the \textit{Look} phase, a robot takes a snapshot of the entire configuration in its own local coordinate system. In the \textit{Compute} phase, it decides either to stay idle or to move to one of its neighboring nodes. In the \textit{Move} phase, it makes an instantaneous move to its computed destination. Based on the timing and activation of the robots, three types of schedulers are common in the literature. In the \textit{fully synchronous} (FSYNC) setting, all the robots are activated simultaneously. The activation phase of all the robots can be divided into global rounds. In the \textit{semi-synchronous} (SSYNC) setting, a subset of robots are activated simultaneously, i.e., not all the robots are necessarily activated in each round. FSYNC can be viewed as a special case of SSYNC. In the \textit{asynchronous} (ASYNC) setting, there is no common notion of time. Moreover, the duration of the \textit{Look, Compute} and \textit{Move} phases is finite but unpredictable and is decided by the adversary for each robot. In this paper, we have considered the scheduler to be \textit{asynchronous}. The scheduler is also assumed to be \textit{fair}, i.e., each robot performs its LCM cycle within finite time and infinitely often.
	
	\noindent In the initial configuration, the robots are placed at the distinct nodes of the grid. The input graph also consists of some \textit{meeting nodes} which are located on the distinct nodes of the grid graph. The \textit{meeting nodes} are visible to the robots during the \textit{look} phase, and they occupy distinct nodes of the grid. A robot can move to one of its adjacent nodes along the grid lines. The movement of the robots is assumed to be instantaneous, i.e., they can be seen only on the nodes of the input grid graph. They are equipped with \textit{global strong multiplicity detection capability}, i.e., in the \textit{Look} phase, they can count the exact number of robots occupying each node. In the \textit{global-weak} version, a robot can detect whether a node is occupied by any robot multiplicity. Unlike in the global versions, the local versions refer to the ability of a robot to perceive information about multiplicities concerning the node in which it resides.   
	
\noindent In this paper, the \textit{optimal gathering over Weber meeting nodes} problem has been studied in an infinite grid model. This is a variant of the {\it gathering over meeting nodes in infinite grid} problem, studied by Bhagat et al. \cite{DBLP:conf/caldam/BhagatCDM20,DBLP:journals/corr/abs-2112-06506}. This paper proposes a deterministic distributed algorithm for the problem with $n\geq 7$ asynchronous robots. The objective constraint is to minimize the total number of moves required by the robots in order to accomplish the \textit{gathering}. In this paper, we have considered \textit{Weber meeting nodes} and observed that if the \textit{gathering node} is a \textit{Weber meeting node}, the algorithm is optimal with respect to the total number of moves made by the robots. Moreover, the \textit{Weber meeting node} is not unique in general, even if the robots are non-collinear.
	\subsection{Motivation}
	\textit{Gathering over meeting nodes in infinite grids} problem was studied by Bhagat et al. \cite{DBLP:conf/caldam/BhagatCDM20,DBLP:journals/corr/abs-2112-06506}, where the robots are assumed to be deployed on the nodes of an infinite grid. The main aim of this paper is to study the problem under the optimization constraint that the sum of the distances traveled by the robots is minimized while accomplishing the \textit{gathering} task. In order to complete the \textit{gathering} task, the robots select a unique \textit{meeting node} and move towards it in such a way that the sum of the lengths of the shortest paths from each robot to the selected \textit{meeting node} is minimized. Since the robots are oblivious, the main challenge of designing a deterministic distributed algorithm lies in keeping the selected \textit{meeting node} invariant while the robots move towards it. It is worth noting here that unlike in the continuous domain, where the robots are represented as points in $\mathbb{R}^2$, they are not allowed to perform infinitesimal movements with infinite precision if they are deployed on the nodes of a graph. This motivates us to consider the specified problem in a grid-based terrain where the robots are only allowed to move along the edges of the input grid graph.  
	
	\subsection{Earlier works}
	In the continuous domain, the robots are represented as points in \textit{Euclidean plane} \cite{BHAGAT201650,10.1007/3-540-45061-0_90,DBLP:series/lncs/FlocchiniPS19,DBLP:journals/tcs/FlocchiniPSW05,DBLP:journals/jpdc/PattanayakMRM19}. Unlike in the graph model, the robots are placed on the nodes of an anonymous graph in general. In an anonymous graph, neither the nodes nor the edges of the graph are labeled, and no local memory is available on the nodes of the graph. The \textit{gathering} problem in the discrete domain has been extensively studied in various topolgies like rings \cite{DBLP:journals/tcs/KlasingMP08,DBLP:journals/dc/DAngeloSN14,10.1007/978-3-642-13284-1_9,10.1007/978-3-642-22212-2_14,10.1007/978-3-642-32589-2_48,DBLP:journals/jda/DAngeloSN14,DBLP:journals/algorithmica/DAngeloSNNS15,DBLP:journals/dc/StefanoN17} finite and infinite grids \cite{DBLP:journals/tcs/DAngeloSKN16,DBLP:journals/iandc/StefanoN17}, bipartite graphs \cite{DBLP:journals/tcs/CiceroneSN21}, complete bipartite graphs \cite{DBLP:journals/tcs/CiceroneSN21}, trees \cite{DBLP:journals/tcs/DAngeloSKN16,DBLP:journals/dc/StefanoN17} and hypercubes \cite{10.1007/978-3-030-14094-6_7}.
	
	\noindent Klasing et al. \cite {DBLP:journals/tcs/KlasingMP08} studied the \textit{gathering} problem in an anonymous ring using \textit{global weak-multiplicity detection} capability. They proved that the \textit{gathering} is impossible without the assumption of \textit{multiplicity detection capability}. They proposed a deterministic distributed algorithm in the asynchronous model for \textit{gathering} an odd number of robots. The algorithm also solves the \textit{gathering} problem for an even number of robots when the initial configuration is asymmetric. D'Angelo et al \cite{DBLP:journals/dc/DAngeloSN14} studied the \textit{gathering} problem in an anonymous ring, where the robots have \textit{global weak-multiplicity detection} capability. They proposed a deterministic distributed algorithm that solves the \textit{gathering} task for any initial configuration which is non-periodic and does not contain any edge-edge line of symmetry. Izumi et. al \cite{10.1007/978-3-642-13284-1_9} studied the \textit{gathering} problem in an anonymous ring and proposed a deterministic algorithm using \textit{local weak-multiplicity detection} capability. D'Angelo et al. \cite{DBLP:journals/jda/DAngeloSN14} studied the \textit{gathering} problem on anonymous rings with $6$ robots in the initial configuration. A unified strategy for all the \textit{gatherable} configurations has been provided in this paper. D'Angelo et al. \cite{DBLP:journals/algorithmica/DAngeloSNNS15} studied the \textit{exploration}, \textit{graph searching} and \textit{gathering} problem in an anonymous ring where the initial configuration is aperiodic and asymmetric. 
	
	\noindent D'Angelo et al. \cite{DBLP:journals/tcs/DAngeloSKN16}, studied the {\it gathering} problem on trees and finite grids. They showed that a configuration remains ungatherable if the configuration is periodic and the dimension of the finite grid consists of at least one even side. A configuration remains ungatherable even if it admits reflection symmetry with the reflection axis passing through the edges of the finite grid. The problem was solved for all the other remaining configurations without assuming any \textit{multiplicity detection capability}. Di Stefano et al. \cite{DBLP:journals/dc/StefanoN17}, studied the optimal {\it gathering} of robots in arbitrary graphs. This paper also introduced the concept of \textit{Weber points} \cite{doi:10.1080/0025570X.1969.11975961,Tan2010} in graphs. A \textit{Weber point} of a graph is a node of a graph that minimizes the sum of the distances from it to each robot. They proposed deterministic algorithms for the \textit{gathering} task on tree and ring topologies that always achieve \textit{optimal gathering} unless the initial configuration is ungatherable. In \cite{DBLP:journals/iandc/StefanoN17}, the optimal \textit{gathering} problem in an infinite grid model was studied by Di Stefano et al. They proposed a deterministic distributed algorithm that minimizes the total distance traveled by all the robots. They proved that their assumed model represents the minimal setting to ensure \textit{optimal gathering}.
	
	\noindent Cicerone et al. \cite{DBLP:journals/tcs/CiceroneSN21} studied the \textit{gathering} problem in arbitrary graphs and proposed a necessary and sufficient result for the feasibility of \textit{gathering} tasks in arbitrary graphs. They have also considered dense and symmetric graphs, like complete and complete bipartite graphs. A deterministic algorithm was proposed that fully characterize the solvability of the \textit{gathering} task in the synchronous setting. Bose et al. \cite{10.1007/978-3-030-14094-6_7} investigated the \textit{optimal gathering} problem in hypercubes, where the optimal criterion is to minimize the total distance traveled by each robot.
	
	\noindent Fujinaga et al. \cite{10.1007/978-3-642-17653-1_1} introduced the concept of \textit{fixed points} or \textit{landmarks} in the \textit{Euclidean plane}. In the \textit{landmarks covering problem}, the robots must reach a configuration, where at each \textit{landmark} point, there is precisely one robot. A distributed algorithm was proposed that assumes common orientation among the robots and minimizes the total number of moves traveled by all the robots. Cicerone et al. \cite{Cicerone2019} studied the \textit{embedded pattern formation (EPF) problem} without assuming common chirality among the robots. The problem asks for a distributed algorithm that requires the robot to occupy all the \textit{fixed points} within a finite amount of time. A variant of the \textit{gathering} problem was studied by Cicerone et al. \cite{CiceroneSN18}, where the \textit{gathering} is accomplished at one of the \textit{meeting points}. These are a finite set of points visible to all the robots during the \textit{Look} phase. They also studied the same problem with respect to the two optimal criteria, one by minimizing the total number of moves traveled by all the robots and the other by minimizing the maximum distance traveled by a single robot. Bhagat et al. \cite{DBLP:conf/caldam/BhagatCDM20,DBLP:journals/corr/abs-2112-06506} studied the \textit{gathering over meeting nodes problem} in an infinite grid. It was shown that even if the robots are endowed with \textit{multiplicity detection capability}, some configurations remain ungatherable. For a given positive integer $k$, the \textit{$k$-circle formation} \cite{a14020062,das2022k} problem asks a set of robots to form disjoint circles having $k$ robots on the circles occupying distinct locations. The circles are centered at the set of fixed points. 
   \subsection{Our contributions}
   This paper proposes a deterministic distributed algorithm for \textit{optimal gathering over Weber meeting nodes} problem, where the initial configurations comprise at least seven robots. The robots are deployed on the nodes of an infinite grid. The optimization criterion considered in this paper is the minimization of the total number of moves made by the robots to finalize the \textit{gathering}. In this paper, a \textit{meeting node} that minimizes the sum of the distances from all the robots is defined as a \textit{Weber meeting node}. Di Stefano et al. \cite{DBLP:journals/dc/StefanoN17} proved that to ensure {\it gathering} by minimizing the total number of moves, the robots must gather at one of the \textit{Weber points}. In our restricted \textit{gathering} model, the robots must gather at one of the \textit{Weber meeting nodes} to ensure {\it gathering} with a minimum number of moves. In this paper, we have shown that there exist some configurations where \textit{gathering over Weber meeting nodes} cannot be ensured, even if the robots are endowed with the \textit{multiplicity detection capability}. This includes the following collection of configurations: 
	\begin{enumerate}
	    \item Configurations admitting a single line of symmetry without any robots or {\it Weber meeting nodes} on the reflection axis. 
	    \item Configurations admitting rotational symmetry without a robot or a {\it meeting node} on the center of rotation. 
	\end{enumerate} 
	In this paper, the assumption on the \textit{multiplicity detection} refers to the \textit{global strong multiplicity detection} capability. We have shown that, without such an assumption, there are configurations where \textit{gathering} cannot be accomplished as soon as a multiplicity is created. However, there are initial configurations where \textit{gathering} can be ensured over a \textit{meeting node}, but not on the set of \textit{Weber meeting nodes}. This includes the configuration admitting a single line of symmetry without any robots or {\it Weber meeting nodes} on the reflection axis, but at least one \textit{meeting node} exists on the reflection axis. In that case, the feasibility of \textit{gathering over meeting nodes} has been studied.
  \subsection{Outline}
     The following section describes the robot model, and the notations used in the paper. Section \ref{s3} provides the formal definition of the problem and the impossibility results for the solvability of the \textit{gathering} task. Section \ref{s4} proposes a deterministic distributed algorithm to solve the \textit{optimal gathering over Weber meeting nodes} problem. Section \ref{s5} describes the correctness of the proposed algorithm. Section \ref{s6} discusses the optimal gathering for the configurations where \textit{gathering over a meeting node} can be ensured, but cannot be ensured over a \textit{Weber meeting node}. Finally, in Section \ref{s7}, we conclude the paper with some discussion for future research.
\section{Optimal Gathering over Weber Meeting Nodes} \label{s2}
The \textit{optimal gathering over meeting nodes} problem has been considered in an infinite grid graph. The objective is to minimize the total distance traveled by all the robots. In order to ensure \textit{optimal gathering over a Weber meeting node}, the robots must finalize {\it gathering} over a {\it meeting node} that minimizes the total distance traveled by all the robots, i.e., each robot must gather at one of the {\it Weber meeting nodes}.
\subsection{Terminology} 
In this subsection, some terminologies and definitions have been proposed.
\begin{itemize}
    \item \textbf{System Configuration}: 
  Let $P=(\mathbb{Z}$, $E')$ denote the infinite path graph, where the vertex set corresponds to the set of integers $\mathbb {Z}$ and the edge set is defined by the ordered pair $E'=\lbrace (i$, $i+1)|$ $i\in \mathbb{Z}\rbrace$. The input grid graph is defined as the \textit{Cartesian Product} of the graph $P\times P$. Let $V$ and $E$ denote the set of nodes and edges of the input grid graph, respectively. $d(u,v)$ denote the lengths of the shortest paths between the nodes $u$ and $v$. Let $M=\lbrace m_1,m_2,\ldots,m_s\rbrace$ denote the finite set of {\it meeting nodes} located on the nodes of the grid graph. $R(t)=\lbrace r_1(t),r_2(t),\ldots,r_n(t)\rbrace$ denote the finite set of robots deploying on the nodes of the grid at any time $t$ respectively. $C(t)=(R(t)$, $M)$ represents the \textit{system configuration} at time $t$. 
	\item \textbf{Symmetry} \cite{DBLP:conf/caldam/BhagatCDM20,DBLP:journals/corr/abs-2112-06506}: An \textit{automorphism} of a graph $G=(V$, $E)$ is a bijective map $\phi:V\rightarrow V$ such that $u$ and $v$ are adjacent if and only if $\phi (u)$ and $\phi (v)$ are adjacent. $f_t:V\rightarrow\lbrace 0, 1  \rbrace$ at any time $t\geq 0$ is defined as follows:
		\[ f_t(v)=\begin{cases} 
		0 & \text{if} \;v \textrm {  is not a meeting node} \\
		1 & \text{if} \;v \textrm{ is a {\it meeting node}} \\
	
		\end{cases}
		\]
		\noindent Let $\lambda_t: V\rightarrow N$ be a function denoting the number of robots on each node $v \in V$ at any time $t\geq0$. An \textit{automorphism} of a configuration $(C(t)$, $f_t, \lambda_t)$ is an \textit{automorphism} $\phi$ of the input grid graph such that $f_t(v)=f_t(\phi(v))$ and $\lambda_t(v)=\lambda_t(\phi(v))$ for all $v\in V$. The grid graph is embedded in the \textit{Cartesian} plane. As a result, a grid can admit only three types of automorphisms: translation, reflection and rotation, and compositions of them. Since the number of robots and {\it meeting node} is finite, a translational automorphism is not possible. An axis of reflection defines a reflection automorphism, while the center of rotation and the angle of rotation determine a rotational automorphism. If the configuration admits reflectional symmetry, the axis of reflection can be horizontal, vertical, or diagonal. The axis of symmetry can pass through the nodes or edges of the graph. In the case of rotational symmetry, the angle of rotation can be $90^{\circ}$ or $180^{\circ}$. The center of rotation can be a node, a center of an edge, or the center of a unit square.
		\begin{figure}[]
		\centering
			\includegraphics[width=0.70\columnwidth]{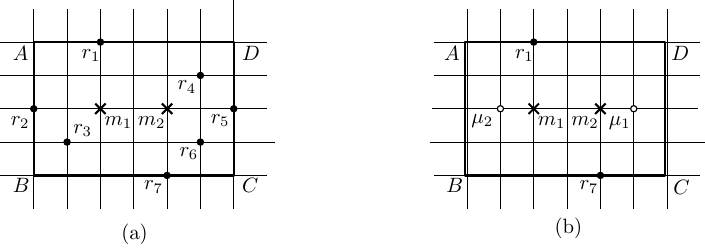}
			\caption{(a) In this figure, cross represents {\it meeting node} and black circles represent robot positions. $m_2$ is the unique \textit{Weber meeting node}. (b) Robots $r_2$ and $r_3$ move towards $m_2$ and create a multiplicity $\mu_2$. Similarly, $r_4$, $r_5$ and $r_6$ move towards $m_2$ and creates a multiplicity $\mu_1$. $m_2$ remains the unique \textit{Weber meeting node} but robots will not be able to compute it correctly if they do not have \textit{global strong multiplicity detection} capability.}
		\label{strong}
	\end{figure} 	
    \item \textbf{Weber meeting node}: Since in the initial configuration the robots are deployed at the distinct nodes of the grid graph, $\lambda_t(v)\leq 1$, $\forall v \in V$. In the final configuration, all the robots are on a single \textit{meeting node} $m \in M$. For a configuration to be final, there must exist a $m \in M$ such that $\lambda_t (m)= n$ and $\lambda_t (v)= 0 $ for each $v \in V \setminus \lbrace m \rbrace $. The \textit{consistency} of a node $m \in M$ at any time $t$ is defined as $c_t(m)= \sum\limits_{v\in V} d(v,m)$ $\lambda_t(v)$. A node $m \in M$ is defined as a {\it Weber meeting node} if it minimizes the value $c_t(m)$. In other words, a \textit{Weber meeting node} $m$ is defined as the \textit{meeting node} which minimizes the sum of the distances from all the robots to itself. The {\it Weber meeting node} may not be unique in general. Let $W(t)$ denote the set of all the {\it Weber meeting nodes} at some time $t$. A deterministic distributed algorithm that gathers all the robots in a \textit{Weber meeting node} via the shortest paths will be optimal with respect to the total number of moves made by the robots. 
	  
\noindent The robots are equipped with \textit{global strong multiplicity detection capability}, i.e., they can detect the exact number of robots occupying any node. Without this assumption, the \textit{Weber meeting nodes} cannot be detected correctly by the robots as soon as a multiplicity is created. As a result, the total number of moves made by the robots to accomplish the \textit{gathering} might not be optimized. For example, consider the configuration in Figure~\ref{strong}(a). If the robots compute the \textit{Weber meeting node} in this configuration, a unique \textit{Weber meeting node} would be computed. Due to the robot's movement, if the configuration in Figure~\ref{strong}(b) is reached, then without the \textit{global strong multiplicity detection capability}, robots will not be able to compute the unique \textit{Weber meeting node} correctly. This example shows that without assuming \textit{strong multiplicity detection capability}, the \textit{gathering over Weber meeting nodes} is no longer possible as soon as a multiplicity is created.
	\begin{figure}[]
			\centering
			{
				
				\includegraphics[width=0.250\columnwidth]{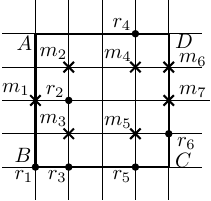}
			}
			
			\caption{ $ABCD$ denotes the minimum enclosing rectangle $M E R $. The lexicographic minimum string in this figure is $s_{DC}$= (1, 2, 6, 8, 16, 18, 22). Here $\alpha_{DC}$= ((0, 0), (1, 0), (1, 0), (0, 1), (0, 0), (0, 1), (1, 0), (0, 0), (1,0), (0, 1), (0, 0), (0, 0), (0, 0), (0, 0), (0, 0), (0, 0), (1, 0), (0, 1), (1, 0), (0, 1), (0, 0), (0, 0), (1, 0), (0, 0),(0, 1)) }
			\label{string}
		\end{figure}
	\item \textbf{Minimum Enclosing Rectangle}: Let $M E R$ = $ABCD$ denote the minimum enclosing rectangle of $R \cup M$. Let the dimension of $M E R $ be $p\times q$, where $|AB|=p$ and $|AD|=q$. The side length $|AB|$ of $MER$ is defined as the number of grid edges on it. Similarly, the other side lengths are defined. For a corner $A$, consider the two possible strings, which are denoted by $s_{AD}$ and $s_{AB}$. The string $s_{AD}$ is defined as follows: Starting from the corner $A$, scan the entire grid by proceeding along the respective direction and sequentially considering all the grid lines parallel to $AD$. While scanning the grid, let $(m_1,m_2,\ldots,m_s)$ denote the list of \textit{meeting nodes}, in the order of their appearance. Let $d_i$ denote the number of hops from the corner $A$ to $m_i$, if measured along the direction parallel to $s_{AD}$. The string $s_{AD}$ is defined by $s_{AD}=(d_1,d_2,\ldots,d_s)$ (Figure \ref{string}). The string $s_{AB}$ is defined similarly. Consider the two possible strings for each corner. Thus, there are a total of eight strings of distances of length $s$ that are obtained by traversing $M E R $. If $M E R $ is a non-square rectangle with $p>q$, the string $s_{AD}$ or $s_{AB}$ is associated for the corner $A$ in the direction of the smallest side ($AD$ in this case). If the \textit{meeting nodes} are asymmetric, there exists a unique string, which is lexicographically minimum among all the possible strings ($s_{DC}$ in Figure \ref{string}). Otherwise, if any two possible strings are equal, then the \textit{meeting nodes} are symmetric. The corner associated with the minimum lexicographic string is defined as a \textit{leading corner}, and the string associated with the \textit{leading corner} is defined as the \textit{string direction} for the respective corner. If $M E R $ is a square, consider the two strings associated with a corner. The string which is lexicographically smaller between the two strings is selected as the \textit{string direction} for the respective corner. The two strings for a respective corner are equal if the \textit{meeting nodes} are symmetric with respect to the diagonal passing through that corner of $M E R $. If all the robots and \textit{meeting nodes} lie on a single line, then $MER$ is a $p \times 1$ rectangle with $A=D$ and $B=C$, and length of $AD$ and $BC$ is 1. Note that, in this case, $s_{AD}$ and $s_{DA}$ refer to the same string. The \textit{meeting nodes} are symmetric when the strings $s_{AD}$ and $s_{DA}$ are equal. 

   		\item \textbf{Potential Weber meeting nodes}: In general, the \textit{Weber meeting nodes} in an infinite grid is not unique. If it is possible to gather at one of the \textit{Weber meeting nodes}, then all the robots must decide to agree on a common \textit{Weber meeting node} for \textit{gathering}. Depending on the symmetricity of the set $M$, the number of \textit{leading corners} is $1, 2$ or $4$ respectively. Consider the \textit{Weber meeting nodes} that represents the last \textit{Weber meeting nodes} in the \textit{string directions} associated to the \textit{leading corners}. Note that the number of \textit{Potential Weber meeting nodes} can be at most eight. Let $W_p(t)$ denote the set of such \textit{Weber meeting nodes} at time $t \geq 0$. The set $W_p(t)$ is defined as the set of \textit{Potential Weber meeting nodes}.
	\item \textbf{Key corner}: Consider all the \textit{leading corners} of $M E R $ and the strings $s_{i}$ associated with each \textit{leading corner} $i$. Assume that there exists at least two \textit{leading corners}. Without loss of generality, assume that $A$ and $D$ are the \textit{leading corners} and the strings parallel to $AD$ and $DA$ are the \textit{string directions} associated with the \textit{leading corners}. The string $\alpha_{AD}$ is defined as follows: Starting from the corner $A$, scan the grid along the \textit{string direction} of $A$, i.e., along $AD$ and associate the pair $(f_t(v), \lambda_t(v))$ to each node $v$ (Figure \ref{string}). The string $\alpha_{DA}$ is defined similarly. Consider the strings $\alpha_{AD}$ and $\alpha_{DA}$. If $C(t)$ is asymmetric, there always exists a unique string which is lexicographically smaller between $\alpha_{AD}$ and $\alpha_{DA}$. If $\alpha _{AD}$ is lexicographically smaller than $\alpha _{DA}$, then the corner $A$ is defined as the \textit{key corner}. If $C(t)$ is symmetric, there may exist more than one \textit{key corner}. Similarly, the strings $\beta_{i}$, for each \textit{non-leading corner} $i$ is defined.
\end{itemize}
In \cite{DBLP:journals/dc/StefanoN17}, it was proved that a \textit{Weber point} remains invariant under the movement of a robot towards itself. In our restricted \textit{gathering} model, where \textit{gathering} can be finalized only on \textit{meeting nodes}, we have the following lemma.
\begin{lemma}\label{wbn}
		Let $m$ be a \textit{Weber meeting node} in a given configuration $C(t)$. Suppose $C(t')$ denotes the configuration after a single robot or a robot multiplicity moves towards the \textit{Weber meeting node} $m$. Then the following results hold.
		\begin{enumerate}
			\item $m\in W(t')$
			\item $W(t')\subseteq W(t)$
			\end{enumerate}
		\label{m1}
	\end{lemma}
	\begin{proof}
	\begin{enumerate}
	\item By definition we have, $c_t(m)= \sum\limits_{v\in V}d(v,m) \lambda_t(v)$ and $W(t)=\lbrace m | \min\limits_{m\in M} c_t(m)\rbrace$. Suppose $r(t)=a$ and $r(t')=b$, i.e., $r$ has moved from the vertex $a$ to $b$ along a shortest path towards $m$ in the time interval $[t,t']$. After the movement of the robot $r$, $\lambda_{t'}(a)$ and $\lambda_{t'}(b)$ becomes $\lambda_t(a)$ - 1 and $\lambda_t(b)$ + 1, respectively. Since $b$ lies on the shortest path from $r$ to $m$ and $d(a,b)=1$, $c_{t'}(m)$ became $c_t(m)$ + $d(b, m)$ - $d(a, m)$, which is again equivalent to $c_t(m)$ - 1. Therefore, $\min\limits_{m\in M} c_{t'}(m)=\min\limits_{m\in M} c_t(m)$ - 1 and hence $m\in W(t')$. Similarly, if a robot multiplicity moves from some vertex $a$ to $b$ at time $t'$ via any shortest path towards $m$, then after the movement of the robot multiplicity, $\lambda_{t'}(a)$ and $\lambda_{t'}(b)$ becomes $\lambda_t(a)$ - $j$ and $\lambda_t(b)+j$, respectively, where $j\geq2$ denotes the number of robots that move from node $a$ to $b$. This implies that, $c_{t'}(m)= c_t(m)$- $j$ and hence, $\min\limits_{m\in M} c_{t'}(m)=\min\limits_{m\in M} c_t(m)$- $j$. Therefore, the \textit{Weber meeting nodes} in the new configuration $C(t')$ are the \textit{Weber meeting nodes} of $C(t)$ which are on some shortest path from $r$ to $m$. Hence $m\in W(t')$.
		
		\item Assume that $m\in W(t')$. This implies that $m$ minimizes the value $c_t'(m)$. The first part of the proof implies that $\min\limits_{m\in M} c_{t'}(m)=\min\limits_{m\in M} c_t(m)$- $j$, where $j \geq 1$ denotes the number of robots that move from node $a$ to $b$. In other words, no node can become a \textit{Weber meeting node} if it was not before the move. Therefore, $m$ must belong to $W(t)$ and hence $W(t')\subseteq W(t)$.
	\end{enumerate}
	\end{proof}
		\begin{figure}[h]
		\centering
		{
			\includegraphics[width=0.250\columnwidth]{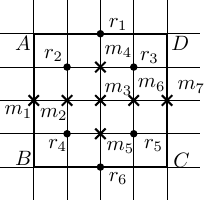}
		}
		
		\caption { Multiple \textit{Weber meeting nodes} in a configuration.}
		\label{weber}
	\end{figure}
	This lemma proves that the \textit{Weber meeting node} remains invariant under the movement of robots towards itself via a shortest path. In Figure \ref{weber}, the configuration admits rotational symmetry. The \textit{Weber meeting node} is not unique. There are three \textit{Weber meeting nodes} $m_3, m_4$ and $m_5$ in the configuration.
	\begin{observation}
		Let $C(0)$ be any initial configuration that admits rotational symmetry. Assume that the center of the rotational symmetry $c$ contains a \textit{meeting node} $m$. Then $m$ is a \textit{Weber meeting node}.
		\label{m3}  
	\end{observation}
\section {Problem Definition and Impossibility Results} \label{s3}
In this section, we have formally defined the problem. A partitioning of the initial configurations has also been provided in this section. 
\subsection{Problem Definition}
Let $C(t)=(R(t)$, $M)$ be a given configuration. The goal of the {\it optimal gathering over Weber meeting nodes problem} is to finalize the \textit{gathering} at one of the {\it Weber meeting nodes} of $C(0)$. We have proposed a deterministic distributed algorithm that ensures \textit{gathering over a Weber meeting node}, where the initial configuration consists of at least seven robots. If $\vert W(t)\vert =1$, then all the robots finalize the \textit{gathering} at the unique \textit{Weber meeting node}. Otherwise, all the robots must agree on a common \textit{Weber meeting node} and finalize the \textit{gathering}.
	\subsection{Partitioning of the Initial Configurations}
	All the initial configurations can be partitioned into the following disjoint classes.
		\begin{enumerate}
			\item  $\mathcal{I}_1-$: Any configuration for which $|W(t)|=1$ (Figure~\ref{initial1}(a)).
			\item  $\mathcal{I}_2-$: Any configuration for which $M$ is asymmetric and $|W(t)|\geq 2$ (Figure~\ref{initial1}(b)).
			\item $\mathcal{I}_3-$  Any configuration for which $M$ admits a unique line of symmetry $l$ and $|W(t)|\geq 2$. This can be further partitioned into:
			\begin{enumerate}
				\item $\mathcal{I}_{3}^a-$ $C(t)$ is asymmetric (Figure~\ref{initial2}(a)).
				\item $\mathcal{I}_{3}^b-$ $C(t)$ is symmetric with respect to the line $l$. This can be further partitioned into: $\mathcal{I}_{3}^{b1}-$ there exists at least one {\it Weber meeting node} on $l$. (Figure~\ref{initial2}(b)), $\mathcal{I}_{3}^{b2}-$ there exists at least one robot position on $l$ but no {\it Weber meeting nodes} on $l$ (Figure~\ref{initial2}(c)), $\mathcal{I}_{3}^{b3}-$ there does not exist any {\it Weber meeting node} or robot position on $l$, but there may exist a \textit{meeting node} on $l$ (Figure~\ref{initial5}(a))  and $\mathcal{I}_{3}^{b4}-$ there does not exist any {\it meeting node} or robot position on $l$.
			
					\begin{figure}[h]
			\centering
			{
				\includegraphics[width=0.250\columnwidth]{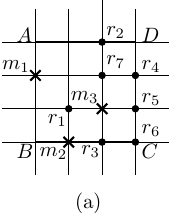}
			}
			\hspace*{1cm}
			{
				\includegraphics[width=0.250\columnwidth]{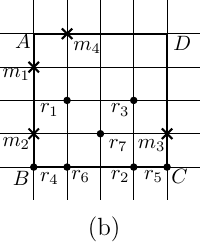}
			}
			\caption{(a) $\mathcal{I}_1$ configuration. $m_3$ is the unique \textit{Weber meeting node}. (b) $\mathcal{I}_2$ configuration. $m_2$ and $m_3$ are the \textit{Weber meeting nodes}. $A$ is the \textit{leading corner}. } 
			\label{initial1}
		\end{figure}
			\end{enumerate}
			\begin{figure}[h]
			\centering
			{
				\includegraphics[width=0.230\columnwidth]{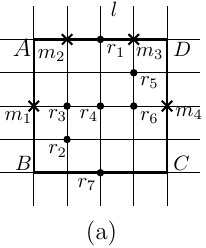}
			}
			\hspace*{0.21cm}
			{
				\includegraphics[width=0.230\columnwidth]{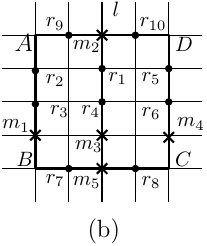}
			}
			\hspace*{0.21cm}
			{
				\includegraphics[width=0.230\columnwidth]{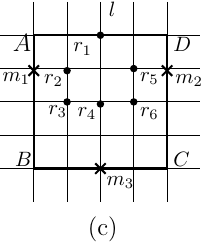}
			}
			\caption{(a) $\mathcal{I}_3^{a}$ configuration. $A$ is the \textit{key corner}. $m_1$ and $m_4$ are the \textit{Weber meeting nodes}. (b) $\mathcal{I}_3^{b1}$ configuration. $m_2$ and $m_3$ are the \textit{Weber meeting nodes}. (c) $\mathcal{I}_3^{b2}$ configuration. $m_3$ is a \textit{meeting node} on $l$, but not a {\it Weber meeting node} on $l$. $r_1$ and $r_4$ are the robot positions on $l$.}
			\label{initial2}
		\end{figure}
				\begin{figure}[h]
			\centering
			{
				\includegraphics[width=0.230\columnwidth]{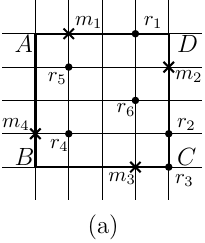}
			}
			\hspace*{0.25cm}
			{
				\includegraphics[width=0.230\columnwidth]{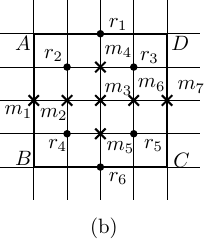}
			}
			\hspace*{0.25cm}
			{\includegraphics[width=0.230\columnwidth]{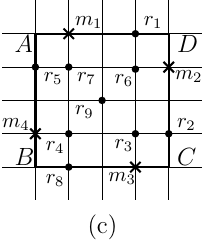}
				
			}
			\caption{(a) $\mathcal{I}_4^{a}$ configuration. $m_2$ and $m_3$ are the \textit{Weber meeting nodes}. (b) $\mathcal{I}_4^{b1}$ configuration. $m_3$, $m_4$ and $m_5$ are the \textit{Weber meeting nodes}. (c) $\mathcal{I}_4^{b2}$ configuration. $m_1$, $m_2$, $m_3$ and $m_4$ are the \textit{Weber meeting nodes}. Robot $r_9$ on the center of rotation.}
			\label{initial3}
		\end{figure}
			\begin{figure}[h]
			\centering
			{
				\includegraphics[width=0.230\columnwidth]{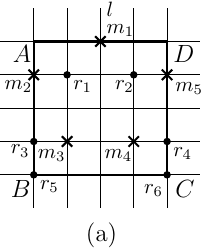}
			}
			\hspace*{0.41cm}
			{
				\includegraphics[width=0.230\columnwidth]{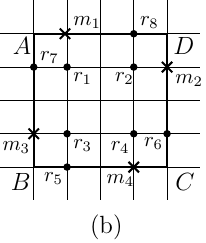}
			}
			
			\caption{(a) $\mathcal{I}_3^{b3}$ configuration. It contains a \textit{meeting node} $m_1$ on $l$ but it is not a \textit{Weber meeting node}. $m_3$ and $m_4$ are the \textit{Weber meeting nodes}. (b) $\mathcal{I}_4^{b3}$ configuration without robots or \textit{meeting nodes} on $c$.} 
			\label{initial5}
		\end{figure}
			\item $\mathcal{I}_4-$  Any configuration for which $M$ admits rotational symmetry with center of rotation $c$ and $|W(t)|\geq 2$. This can be further partitioned into:
			\begin{enumerate}
				\item $\mathcal{I}_{4}^a-$ $C(t)$ is asymmetric (Figure~\ref{initial3}(a)).
				\item $\mathcal{I}_{4}^b-$ $C(t)$ is symmetric with respect to rotational symmetry or $C(t)$ may admit a single line of symmetry. This can be further partitioned into: $\mathcal{I}_{4}^{b1}-$There exists a {\it meeting node} on $c$, $\mathcal{I}_{4}^{b2}-$ there exists a robot position on $c$ and $\mathcal{I}_{4}^{b3}-$ there does not exist any {\it meeting node} or robot positions on $c$ (Figure~\ref{initial5}(b)), or on any line of symmetry. 
		\end{enumerate}
		\end{enumerate}
We assume that if the \textit{meeting nodes} are symmetric with respect to a single line of symmetry, then $l$ is the line of symmetry. Similarly, if the \textit{meeting nodes} are symmetric with respect to rotational symmetry, then $c$ is the center of rotational symmetry. Since the partitioning of the initial configurations depends only on the position of \textit{meeting nodes}, which are fixed nodes, all the robots can determine the class of configuration in which it belongs without any conflict. Let $\mathcal I$ denote the set of all initial configurations.
	
		\begin{lemma}
		If the initial configuration $C(0) \in \mathcal{I}_{3}^{b3} \cup \mathcal{I}_{4}^{b3}$, then the \textit{gathering} over Weber meeting nodes problem cannot be solved.
		\label{m4}
	\end{lemma}	
	The proof of the above lemma can be observed as a corollary to Theorem 1, proved in Bhagat et al.   \cite{DBLP:conf/caldam/BhagatCDM20,DBLP:journals/corr/abs-2112-06506}. In \cite{DBLP:conf/caldam/BhagatCDM20,DBLP:journals/corr/abs-2112-06506}, it was proved that $\mathcal{I}_{3}^{b4}$ is ungatherable. Let $\mathcal U$ denote the set of all configurations for which \textit{gathering over a Weber meeting node} cannot be ensured. According to Lemma \ref{m4}, this includes all the configurations,
\begin{enumerate}
    \item admitting a single line of symmetry $l$, and $l \cap (R\cup W(t)) = \phi$.
    
    \item admitting rotational symmetry with center $c$ and $\lbrace c \rbrace \cap (R \cup M)= \phi$. 
\end{enumerate}

 \noindent Note that according to Observation \ref{m3}, if $c$ is a \textit{meeting node} on $c$, then it must be a \textit{Weber meeting node}.
 \section{Algorithm} \label{s4}
\subsection{Overview of the Algorithm}
In this subsection, a deterministic distributed algorithm has been proposed to solve the \textit{optimal gathering} problem by \textit{gathering} each robot at one of the \textit{Weber meeting nodes}. The proposed algorithm works for all the configurations $C(t) \in \mathcal I \setminus (\mathcal U \cup \mathcal I_3^{b4})$ consisting of at least seven robots. The main strategy of the algorithm is to select a \textit{Weber meeting node} among all the possible \textit{Potential Weber meeting nodes} and allow the robots to move towards the selected \textit{Weber meeting node}. The proposed algorithm mainly consists of the following phases: \textit{Guard Selection, Target Weber meeting node Selection, Leading Robot Selection, Symmetry Breaking, Creating Multiplicity on Target Weber meeting node} and \textit{Finalisation of Gathering}. In the \textit{Target Weber meeting node Selection} phase, the \textit{Potential Weber meeting node} for \textit{optimal gathering} is selected. The \textit{Weber meeting node} selected for \textit{gathering} is defined as the \textit{target Weber meeting node}. A set of robots denoted as \textit{guards} are selected in the \textit{Guard Selection} phase. \textit{Guards} are selected in order to ensure that the initial $M E R $ remains invariant. In the \textit{Leading Robot Selection} phase, a robot is selected as a \textit{leading robot} and placed. A unique robot is selected and allowed to move towards an adjacent node in the \textit{Symmetry Breaking phase}. This movement of the robot transforms a symmetric configuration into an asymmetric configuration. All the \textit{non-guard} robots move towards the \textit{target Weber meeting node}, thus creating a multiplicity on it in the \textit{Creating Multiplicity on Target Weber meeting node} phase. Finally, all the \textit{guards} move towards the uniquely identifiable (robots have \textit{global strong multiplicity detection capability}) \textit{target Weber meeting node} in the \textit{Finalisation of Gathering} phase and finalize the \textit{gathering}. 
\subsection{Half-planes and Quadrants}
	 Assume that the initial configuration $C(0)$ is asymmetric. First, consider the case when the locations of the \textit{meeting nodes} are symmetric with respect to a single line of symmetry $l$. The line $l$ divides $M E R $ into two half-planes. The half-planes defined in this section are open half-planes, i.e., excluding the nodes on $l$. If the \textit{meeting nodes} are symmetric with respect to rotational symmetry and $c$ is the center of rotation, then consider the lines $l$ and $l'$ which pass through $c$. These lines are perpendicular to each other and divide the $MER$ into four quadrants. The quadrants defined in this section are open quadrants, i.e., the quadrants exclude the nodes belonging to the lines $l$ and $l'$. A configuration is said to be \textit{balanced} if the following conditions hold:
	 \begin{enumerate}
	     \item $C(0) \in \mathcal{I}_{3}^a$ and the half-planes delimited by $l$ contain an equal number of robots.
	     \item $C(0) \in \mathcal{I}_{4}^a$. Assume that there exist at least two quadrants that contain the maximum number of \textit{Potential Weber meeting nodes}. Suppose more than one quadrant contains either the maximum or the minimum number of robots among all the specified quadrants. In that case, the configuration is said to be balanced.
	 \end{enumerate}
	 If the initial configuration is not balanced, then it is an \textit{unbalanced} configuration. 
	 \noindent An initial configuration $C(0)$ satisfies the following conditions:
	 \renewcommand\labelitemi{\tiny$\bullet$}
	
	 \begin{itemize}
	     \item $C_1$: there exists a unique half-plane or quadrant that contains the maximum number of \textit{Potential Weber meeting nodes}.
	     \item $C_2$: there exists multiple half-planes or quadrants that contain the maximum number of \textit{Potential Weber meeting nodes}. Any configuration $C(0)$ satisfying condition $C_2$ is said to satisfy $C_{21}$, if $C(0)$ is balanced. Otherwise, it satisfies $C_{22}$, if the initial configuration is unbalanced.
	     \item $C_3$: there does not exist any \textit{Potential Weber meeting node} on the half-planes or on the quadrants.
	 \end{itemize}
 \begin{figure}[h]
			\centering
			{
				\includegraphics[width=0.26\columnwidth]{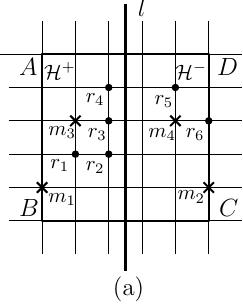}
				}
			\hspace*{0.81cm}
			{
				\includegraphics[width=0.26\columnwidth]{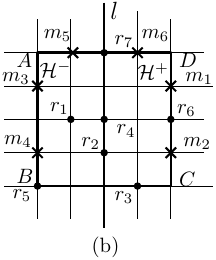}
			}
			\caption{(a) $l$ denotes the single line of symmetry. $m_3$ and $m_4$ are the \textit{Weber meeting nodes}. $\mathcal H^{+}$ is defined as the half-plane with the maximum number of robots. (b) $A$ and $D$ are the \textit{leading corners}. $A$ is the \textit{key corner}. $m_2$ and $m_4$ are the \textit{Weber meeting nodes}. $\mathcal H^{+}$ is defined as the half-plane not containing the \textit{key corner} $A$.}
			\label{half1}
		\end{figure}	
\begin{table}
\centering
\begin{tabular}{ |p{5cm}|p{7cm}|  }
\hline
\multicolumn{2}{|c|}{\textbf{Demarcation of the half-planes for fixing the target}} \\
\hline
\multicolumn{1}{|c|}{\textbf{Initial Configuration} $\boldsymbol{C(0)}$} & \multicolumn{1}{|c|}{$\boldsymbol{\mathcal H^{+}}$}\\ 
\hline
{\small satisfy $C_1$ }  & T{\small he unique half-plane containing the \textit{Potential Weber meeting nodes}  } \\
\hline
{\small satisfy $C_{21}$ $\land$ $l$ is a horizontal or vertical line of symmetry} & {\small The unique half-plane not containing the \textit{key corner} }\\ 
\hline
{\small satisfy $C_{21}$ $\land$ $l$ is a diagonal line of symmetry $\land$ $\exists$ a unique \textit{leading corner}  }  & {\small The half-plane which lies in the direction of $AD$, if $\alpha_{AD}$ is lexicographically larger than $\alpha_{AB}$ }\\
\hline
{\small satisfy $C_{21}$ $\land$ $l$ is a diagonal line of symmetry $\land$ $\exists$ two \textit{leading corners}}  & {\small The half-plane containing the corners $A$ and $D$, if $\alpha_{AD}$ is lexicographically larger than $\alpha_{CD}$} \\
\hline
{\small satisfy $C_{22}$} &  {\small The unique half-plane with the maximum number of robots} \\
\hline
\end{tabular}
\caption{\label{tab:half}Demarcation of the half-planes}
\end{table}
 \begin{figure}[h]
			\centering
			{
				\includegraphics[width=0.70\columnwidth]{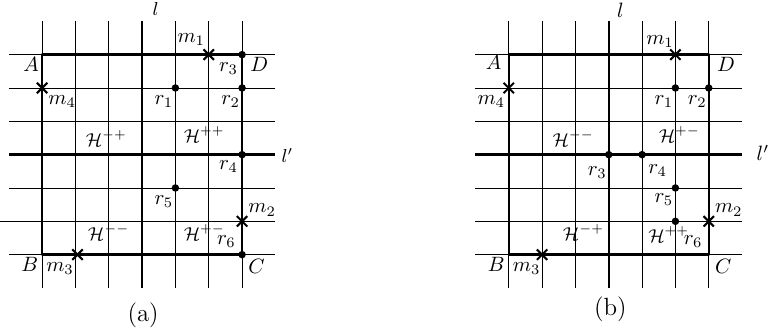}
				
			}
		
			\caption{(a) $m_1$ and $m_2$ are the \textit{Weber meeting nodes}. $\mathcal H^{++}$ denotes the unique quadrant with the maximum number of robots. (b) $m_1$ and $m_2$ are the \textit{Weber meeting nodes}. The quadrants containing the corners $C$ and $D$, contain the maximum number of robots. $D$ is the \textit{key corner}. $\mathcal H^{++}$ denotes the quadrant with the maximum number of robots and not containing the unique \textit{key corner}.}
			\label{half2}
		\end{figure}
\subsubsection{Demarcation of the Half-planes for fixing the target} Assume that the \textit{meeting nodes} are symmetric with respect to a single line of symmetry $l$. Note that $\mid W_p(t)\mid \leq 2$. Further, assume that $\mid W_p(t)\mid = 2$ and $C(0)$ does not satisfy $C_3$. This implies that there exists at least one \textit{Potential Weber meeting node} located on the half-planes. Note that, if $l$ is a diagonal line of symmetry, then there may exist one or two \textit{leading corners}. If there exists a unique \textit{leading corner}, then without loss of generality, let $A$ be the \textit{leading corner}. Otherwise, if there exist two \textit{leading corners}, then assume that $A$ and $C$ are the \textit{leading corners} and the \textit{string directions} associated to the corners $A$ and $C$ are along the sides $AD$ and $CD$, respectively. $\mathcal H^{+}$ is defined according to Table \ref{tab:half}. The other half-plane delimited by $l$ is defined as $\mathcal H^{-}$ (Figure \ref{half1} (a) and \ref{half1} (b)).
\subsubsection{Demarcation of Quadrants for fixing the target} First, consider the case when the \textit{meeting nodes} are symmetric with respect to rotational symmetry without multiple lines of symmetry and $W_p(t)\geq 2$. The quadrant $\mathcal H^{++}$ is defined according to Table \ref{tab:quad}. The other quadrants are defined as follows.
		\renewcommand\labelitemi{\tiny$\bullet$}
		\begin{itemize}
		    \item $\mathcal H^{-+}$:- The quadrant adjacent to $\mathcal H^{++}$ with respect to the line $l$.
		    \item $\mathcal H^{+-}$:- The quadrant adjacent to $\mathcal H^{++}$ with respect to the line $l'$.
		    \item $\mathcal H^{--}$:- The quadrant which is non-adjacent to $\mathcal H^{++}$ (Figure~\ref{half2}(a) and \ref{half2}(b)). 
		\end{itemize}
If $M E R $ is a square, and the configuration admits multiple lines of symmetry, there can be at most four lines of symmetry. If there are more than two lines of symmetry, the two lines that are perpendicular to each other and do not pass through any corner of $M E R $ are selected and considered as $l$ and $l'$. Consider the quadrants delimited by the lines $l$ and $l'$. The quadrants are defined similarly, as in the case when $M E R $ admits rotational symmetry without multiple lines of symmetry.
\begin{table}
\centering
\begin{tabular}{ |p{6cm}|p{7cm}|  }
 \hline
 \multicolumn{2}{|c|}{\textbf{Demarcation of the quadrants for fixing the target}} \\
 \hline
 \multicolumn{1}{|c|}{\textbf{Initial Configuration} $\boldsymbol{C(0)}$} & \multicolumn{1}{|c|}{$\boldsymbol{\mathcal H^{++}}$}\\ 
 \hline
{\small satisfy $C_1$} & {\small The unique quadrant containing the maximum number of \textit{Potential Weber meeting nodes} }  \\
\hline
 {\small satisfy $C_{21}$ $\land$ the angle of rotation is $180^{\circ}$ $\land$ $\exists$ at least one quadrant that contains the \textit{Potential Weber meeting nodes} as well as the \textit{leading corners}} & {\small The unique quadrant containing the \textit{leading corner} with which the largest lexicographic string $\alpha_i$ is associated, and that contains the maximum number of robots }\\ 
 \hline
  {\small satisfy $C_{21}$ $\land$ the angle of rotation is $180^{\circ}$ $\land$ the quadrants that contain the \textit{Potential Weber meeting nodes}, do not contain the \textit{leading corners}} & {\small The unique quadrant containing the \textit{non-leading corner} with which the largest lexicographic string $\beta_i$ is associated, and that contains the maximum number of robots} \\
  \hline
{\small  satisfy $C_{21}$ $\land$ the angle of rotation is $90^{\circ}$} & {\small The quadrant containing the corner with which the largest lexicographic string $\alpha_i$ is associated, and that contains the maximum number of robots } \\
  \hline
 {\small satisfy $C_{22}$ } &   {\small The unique quadrant with the maximum number of robots }\\
  \hline
 {\small satisfy $C_3$ $\land$ unbalanced }& {\small The unique quadrant containing the minimum number of robots }\\
  \hline
 {\small satisfy $C_3$ $\land$ balanced} & {\small The unique quadrant that contains the smallest lexicographic string $\alpha_i$ associated with the \textit{leading corner} and containing the minimum number of robots} \\
 
 \hline
\end{tabular}
\caption{\label{tab:quad}Demarcation of the quadrants.}
\end{table}
		\begin{figure}[h]
			\centering
			{
				\includegraphics[width=0.250\columnwidth]{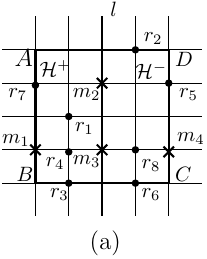}
			}
			\hspace*{0.75cm}
			{
				\includegraphics[width=0.250\columnwidth]{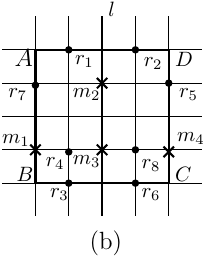}
			}
			
			\caption{(a) \textit{Meeting nodes} are symmetric with respect to a single line of symmetry $l$, but the configuration is asymmetric. $B$ is the \textit{leading corner} contained in $\mathcal H^{+}$. Robots $r_2$ and $r_3$ are selected as \textit{guards}. (b) The configuration is symmetric with respect to a single line of symmetry $l$. $B$ and $C$ are the \textit{leading corners}. Robots $r_1$, $r_2$, $r_3$ and $r_6$ are selected as \textit{guards}. }
				\label{guard}
			\end{figure}
\subsection{Phases of the Algorithm}
The proposed algorithm mainly consists of the following phases.
\subsubsection{Guard Selection}
In this phase, a set of robots is selected as \textit{guards} in order to keep the initial $M E R $ invariant. If there does not exist any \textit{meeting nodes} on a side of the boundary of $M E R $, then there must exist at least one robot on that particular side of the boundary. \textit{Guards} are selected in such a way that they remain uniquely identifiable. If a side of the boundary of $M E R $ contains at least one \textit{meeting node}, then a \textit{guard} robot is not required for that particular side of the boundary. Therefore, consider the case when the boundary of $M E R $ does not contain any \textit{meeting nodes}. Consider the robots which are on the boundary of the $M E R $. First, assume that $C(t)$ is asymmetric. Let $G$ denote the set of \textit{guards}. Let $G_C$ denote the set of \textit{guard corner} and is defined as follows.
	\renewcommand\labelitemi{\tiny$\bullet$}
	\begin{itemize}
	    \item The unique \textit{leading corner}, if the \textit{meeting nodes} are asymmetric.
	    \item The \textit{leading corner} contained in $\mathcal H^{+}$, if the \textit{meeting nodes} are symmetric with respect to a horizontal or vertical single line of symmetry $l$. The unique \textit{key corner} contained in $\mathcal H^{+}$, if the \textit{meeting nodes} are symmetric with respect to a diagonal line of symmetry.
	    \item The \textit{leading corner} contained in $\mathcal H^{++}$, if the \textit{meeting nodes} are symmetric with respect to rotational symmetry. 
	\end{itemize}
	The robot positions on the sides adjacent to the unique \textit{guard corner} and are closest to the \textit{guard corner} are considered as \textit{guards}. Similarly, the robots that are farthest from the \textit{guard corner} measured along the \textit{string direction} and lying on the sides non-adjacent to the \textit{guard corner} are also considered as \textit{guards}. Note that, in each case, there are exactly four \textit{guard} robot positions that are selected in this phase (Figure \ref{guard}(a)).

\noindent If $C(t)$ is symmetric with respect to a unique line of symmetry $l$ and $l$ is a horizontal or vertical line of symmetry, there are exactly two \textit{leading corners}. Consider the robot positions on the sides adjacent to the \textit{leading corners} and which are closest to the \textit{leading corners}. These two robots and their symmetric images are selected as \textit{guards}. The robots which are farthest from the \textit{leading corners} and lying on the side which are non-adjacent to the \textit{leading corners} are also selected as \textit{guards}. Hence, there are exactly six \textit{guard} robots that are selected when $C(t)$ is symmetric with respect to $l$ (Figure \ref{guard}(b)). Otherwise, if $l$ is a diagonal line of symmetry and there exists a unique \textit{leading corner}, then the robots positions on the sides adjacent to the \textit{leading corner} and are closest to the \textit{leading corner} are selected as \textit{guards}. The robot positions on the sides non-adjacent to the \textit{leading corner} and farthest from the \textit{leading corner} are also selected as \textit{guards}. Note that they are symmetric images of each other. If there are two \textit{leading corners}, the robots which are closest to the \textit{leading corners} and lying on the sides adjacent to the \textit{leading corners} are selected as \textit{guards}. Note that, if $C(t)$ is symmetric with respect to rotational symmetry, then since the center of rotational symmetry is also the center of fixed \textit{meeting nodes} and \textit{gathering} is finalized in the center, the \textit{Guard Selection} phase is not executed in this case.
\begin{table}
\centering
\begin{tabular}{ |p{6.5cm}|p{7cm}|  }
 \hline
 \multicolumn{2}{|c|}{\textbf{Target Weber meeting node Selection}} \\
 \hline
 \multicolumn{1}{|c|}{\textbf{Configuration} $\boldsymbol{C(t)}$} & \multicolumn{1}{|c|}{\textbf{Target Weber meeting node}}\\ 
 \hline
{\small  Admitting a unique \textit{Weber meeting node} } & {\small The unique \textit{Weber meeting node}} \\
\hline
  {\small Admitting a unique \textit{Potential Weber meeting node}} & {\small The unique \textit{Potential Weber meeting node}} \\ 
 \hline
  {\small $\mathcal I_3$ $\land$ there exists a \textit{Weber meeting node} on $l$} & {\small The northernmost \textit{Weber meeting node} on $l$ }\\
  \hline
 {\small $\mathcal I_3^{a}$ $\land$ there does not exist any \textit{Weber meeting node} on $l$ $\land$ $|W_p(t) = 2|$ $\land$ $l$ is a horizontal or vertical line of symmetry} & {\small The \textit{Potential Weber meeting node} in $\mathcal H^{+}$. Ties are broken by considering the \textit{Potential Weber meeting node} which appears last in the \textit{string direction} associated to the \textit{leading corner} in $\mathcal H^+$} \\
  \hline
 {\small $\mathcal I_3^a$ $\land$ there does not exist any \textit{Weber meeting node} on $l$ $\land$ $|W_p(t) = 2|$ $\land$ $l$ is a diagonal line of symmetry $\land$ there exists a unique \textit{leading corner}} & {\small The \textit{Potential Weber meeting node} in $\mathcal H^+$ which appears last in the \textit{string direction} associated to the unique \textit{leading corner} }\\
  \hline
  {\small $\mathcal I_3^a$ $\land$ there does not exist any \textit{Weber meeting node} on $l$ $\land$ $|W_p(t) = 2|$ $\land$ $l$ is a diagonal line of symmetry $\land$ there exists two \textit{leading corner}} & {\small The \textit{Potential Weber meeting node} in $\mathcal H^{++}$ that appears last in the \textit{string direction} associated to the \textit{key corner} in $\mathcal H^{+}$}\\
  \hline
{\small $\mathcal I_4$ $\land$ there exists a \textit{Weber meeting node} }on $c$ & {\small The \textit{Weber meeting node} on $c$} \\ 
 
  \hline 
      {\small $\mathcal I_4^a$ $\land$ there does not exist a \textit{Weber meeting node} on $c$ $\land$ $|W_p(t)| \geq 2$ $\land$ there does not exist any \textit{Weber meeting node} on the quadrants} & {\small The \textit{Potential Weber meeting node} which is closest from the unique \textit{key corner} in the \textit{string direction} and lying on either $l$ or $l'$}  \\
  \hline
  {\small $\mathcal I_4^a$ $\land$ there does not exist a \textit{Weber meeting node} on $c$ $\land$ $|W_p(t)| \geq 2$ $\land$ there exists a \textit{Weber meeting node} on the quadrants} & {\small The \textit{Potential Weber meeting node} in $\mathcal H^{++}$ which is farthest from the \textit{leading corner} contained in $\mathcal H^{++}$ in the \textit{string direction}}.   \\
 \hline
\end{tabular}
\caption{\label{tab:target}Target Weber meeting node selection.}
\end{table}
\begin{algorithm2e}[!h]
		\KwIn{$C(t)=(R(t)$, $M)$}
		\footnotesize
		\uIf{$C(t)\in I_1$}
		{Select the unique \textit{Weber meeting node} $m$\;} 
		
		\uElseIf{$C(t)\in I_2$}
		{ Select the unique \textit{Potential Weber meeting node} \;}
		\uElseIf{$C(t)\in I_3$}
		{
			\uIf{$C(t)$ is asymmetric and $l\cap W(t)\neq \phi$}
			{Select the northernmost \textit{Weber meeting node} on $l$}
			\uElseIf{$C(t)$ is asymmetric and $(l \cap W(t) = \phi$ $\land$ $|W_p(t)|=1$)}
			{Select the unique \textit{Potential Weber meeting node} \;}
			\uElseIf{$C(t)$ is asymmetric and ($l \cap W(t) = \phi \land |W_p(t)|=2$)}
			{
			\uIf{there exists a unique \textit{leading corner} in $\mathcal H^{+}$}
			{ Select the \textit{Potential Weber meeting node} in $\mathcal H^+$ which appears last in the \textit{string direction} associated to the unique \textit{leading corner}\;}
			\ElseIf{there exist two \textit{leading corner} in $\mathcal H^{+}$}
			{Select the \textit{Potential Weber meeting node} that appears last in the \textit{string direction} associated to the \textit{key corner} in $\mathcal H^{+}$ \;}
			}

			\ElseIf{$C(t)$ is symmetric with respect to the line $l$}
			{Select the northernmost \textit{Weber meeting node} on $l$ \;}
			
		}
		\ElseIf{$C(t)\in I_4$}
		{
			\uIf{$C(t)$ is asymmetric and $\lbrace c \rbrace\cap W(t) \neq\phi$}
			{Select the \textit{Weber meeting node} on $c$\;}
			\uElseIf{$C(t)$ is asymmetric and ($\lbrace c \rbrace\cap W(t)=\phi \land \mid W_p(t)\mid=1$)}
			{Select the unique \textit{Potential Weber meeting node} \;}
			\uElseIf{$C(t)$ is asymmetric and ($\lbrace c \rbrace\cap W(t)=\phi \land \mid W_p(t)\mid \geq 2$)}
			{\uIf{All the \textit{Potential Weber meeting nodes} lie either on line $l$ or $l'$}
				{Select the \textit{Potential Weber meeting node} which is farthest from the unique \textit{key corner} in the \textit{string direction} and lying on either $l$ or $l'$\;
			}
			\ElseIf{there exists a \textit{Potential Weber meeting node} lying on the quadrants}
			{Select the \textit{Potential Weber meeting node} in $\mathcal H^{++}$ which is farthest from the \textit{leading corner} contained in $\mathcal H^{++}$ in the \textit{string direction} \;}
			}
		   \ElseIf{$C(t)$ is symmetric with a meeting node on the center of rotation $c$}
			{
			Select the \textit{meeting node} on $c$\;
			}
		}
		\caption{Target Weber meeting node Selection()}
		\label{code1}
	\end{algorithm2e}
\subsubsection{Target Weber meeting node Selection}
	In this phase, the \textit{Weber meeting node} for \textit{gathering} is selected. The \textit{target meeting node} must remain invariant during the execution of the algorithm. Depending on the class of configuration to which $C(t)$ belongs, the \textit{target Weber meeting node} is selected according to Table \ref{tab:target}. The pseudo-code corresponding to this phase is given in Algorithm \ref{code1}. Consider the case when the $C(t) \in \mathcal I_4^a$ and there exists a \textit{Weber meeting node} on the quadrants. Further, assume that $|W_p(t)| \geq 2$. If there exist two \textit{string directions} corresponding to the unique \textit{leading corner} in $\mathcal H^{++}$, the \textit{target Weber meeting node} is selected as the \textit{Potential Weber meeting node} in $\mathcal H^{++}$ which appears first in the string $\alpha_i$. We have the following observation.
	\begin{observation} \label{o2}
	If the \textit{meeting nodes} are symmetric with respect to a unique line of symmetry $l$, and there exists at least one \textit{meeting node} on $l$, then the \textit{meeting nodes} on $l$ are orderable.
	\end{observation}
The northernmost \textit{meeting node} on $l$ is defined as the \textit{meeting node} on $l$ which is farthest from the \textit{leading corner(s)}. Similarly, the northernmost robot on $l$ is defined.

		\begin{figure}[h]
			\centering
		
			{
				\includegraphics[width=0.250\columnwidth]{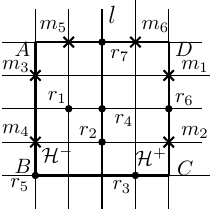}
			}
		
			\caption{Balanced $\mathcal I_3^ a$ configuration. $m_2$ and $m_4$ are the \textit{Weber meeting nodes} lying on different half-planes. $m_2$ is selected as the \textit{target Weber meeting node}. Robot $r_2$ is selected as the \textit{leading robot}}.
			\label{leading}
		\end{figure}
\subsubsection {Leading Robot Selection} 
If the initial configuration is balanced and asymmetric, a robot $r$ is selected as a \textit{leading robot} in the \textit{Leading Robot Selection} phase (Figure \ref{leading}). The \textit{leading robot} moves towards the half-plane or quadrant containing the \textit{target Weber meeting node} $m$. While $r$ reaches the half-plane or the quadrant containing $m$, the configuration transforms into an unbalanced configuration, and the asymmetry of the configuration remains invariant. Since the initial configuration is balanced, assume that $C(t)\in \mathcal I_3 \cup \mathcal I_4$. Further, assume that the initial configuration does not satisfy the condition $C_3$. Depending on the class of configuration to which $C(t)$ belongs, the \textit{leading robot} is selected according to Table \ref{tab:lead}. In case, the configuration is in  $\mathcal I_4^a$ and there exists a robot on $l$ (resp. $l'$), the \textit{leading robot} first move along the line $l$ (resp. $l'$) and when it becomes collinear with $m$, it starts moving along $l'$ (resp. $l$).
\begin{table}
\centering
\begin{tabular}{ |p{6.5cm}|p{7cm}|  }
 \hline
 \multicolumn{2}{|c|}{\textbf{Leading Robot Selection}} \\
 \hline
\textbf{ Configuration} $\boldsymbol{C(t)}$ & \textbf{Leading Robot} \\ 
 \hline
 {\small $\mathcal I_3^a$ $\land$ there exists a robot on $l$} & {\small The northernmost robot on $l$} \\
\hline
  {\small$\mathcal I_{3}^a$ $\land$ there does not exist any robot on $l$} & {\small The robot closest to $l$ and lying on $\mathcal H^{-}$. Ties are broken by considering the robot on $\mathcal H^{-}$ which is farthest from the \textit{leading corner} contained in $\mathcal H^{-}$ in the \textit{string direction}} \\ 
 \hline
  {\small $\mathcal I_4^a$ $\land$ there exists a robot either on $l$ or $l'$} & {\small The robot closest to the \textit{target Weber meeting node} and lying on $l$ or $l'$. Ties are broken by considering the robot either on $l$ or $l'$ which is closest from the \textit{leading corner} contained in $H^{++}$ in the \textit{string direction}} \\
  \hline
{\small $\mathcal I_4$ $\land$ there does not exist any robot on $l$ and $l'$ $\land$ there exists a \textit{non-guard} robot in a quadrant adjacent to $\mathcal H^{++}$} & {\small The robot lying on a quadrant adjacent to $\mathcal H^{++}$ and closest to the \textit{target Weber meeting node}. Ties are broken by considering the robot, which is closest from the \textit{leading corner} contained in $H^{++}$ in the \textit{string direction}}\\
  \hline
 {\small $\mathcal I_4$ $\land$ there does not exist any robot on $l$ and $l'$ $\land$ there does not exist any \textit{non-guard} robot in the quadrants adjacent to $\mathcal H^{++}$} & {\small The robot lying on the quadrant non-adjacent to $\mathcal H^{++}$ and closest to the \textit{target Weber meeting node}. Ties are broken by considering the robot, which is closest from the \textit{leading corner} contained in $H^{++}$ in the \textit{string direction}}\\
  \hline
\end{tabular}
\caption{\label{tab:lead}Leading Robot Selection.}
\end{table}

\subsubsection{Symmetry Breaking}
 In this phase, all the symmetric configurations that can be transformed into asymmetric configurations are considered. A unique robot is identified that allows the transformation. We have the following cases.
 \begin{enumerate}
     \item $C(t)\in\mathcal{I}_{3}^{b2}$. In this class of configurations, at least one robot exists on $l$. Let $r$ be the northernmost robot on $l$. $r$ moves towards an adjacent node that does not belong to $l$, and the configuration becomes asymmetric. 
     
     \item $C(t)\in\mathcal{I}_{4}^{b2}$. In this class of configurations, there exists a robot (say $r$) on $c$. The robot $r$ moves towards an adjacent node. If the configuration admits rotational symmetry with multiple lines of symmetry and there is a robot $r$ at the center, $r$ moves towards an adjacent node. This movement creates a unique line of symmetry $l'$. However, the new position of $r$ might have a multiplicity. If that happens to be the northernmost robot on $l'$, moving robots from there might still result in a configuration with a line of symmetry. Even so, the unique line of symmetry $l'$ would still contain at least one robot position without multiplicity, and the number of robot positions on $l'$ will be strictly less than the number of robots on the line of symmetry in the original configuration. Thus, the repeated movement of the robot on $l'$ guarantees to transform the configuration into an asymmetric configuration. 
 \end{enumerate}
\subsubsection {Creating Multiplicity on Target Weber meeting node}
	The \textit{target Weber meeting node} $m$ is selected in the \textit{Target Weber meeting node Selection} phase. Since there is a unique \textit{target Weber meeting node} $m$, all the \textit{non-guard} robots move towards $m$ in the \textit{Creating Multiplicity on Target Weber meeting node} phase. Note that, since the \textit{guards} do not move during this phase, the $M E R$ remains invariant. As a result, $m$ remains invariant. Eventually, a robot multiplicity is created on $m$, while the \textit{non-guards} moves towards it. Depending on the class of configuration to which $C(t)$ belongs, the following cases are to be considered.
	\begin{enumerate}
	    \item $C(t) \in \mathcal I_1:$ All the robots moves towards the unique \textit{Weber meeting node} $m$.
	    \item $C(t) \in \mathcal I_2:$ All the \textit{non-guards} move towards the unique \textit{target meeting node} $m$.
	    \item $C(t) \in \mathcal I_3:$ If $C(t) \in \mathcal I_3^a$ and there exists a \textit{Weber meeting node} on $l$, each \textit{non-guards} move towards the \textit{target meeting node} $m$. 
	    
	    \noindent Next, consider the case when $C(t) \in \mathcal I_3^a$ and there does not exist any \textit{Weber meeting node} on $l$. A \textit{leading robot} in the \textit{Leading Robot Selection} phase transforms a balanced configuration into an unbalanced configuration. All the \textit{non-guard} robots from $\mathcal H^{-}$ move towards $m$. This movement is required in order to ensure that $\mathcal H^{+}$ remains invariant. While such robots reach $\mathcal H^{+}$, all the \textit{non-guard} robots in $\mathcal H^{+}$ move towards $m$, thus creating a multiplicity on $m$. 
	    
	    \noindent Finally, if $C(t) \in \mathcal I_3^{b2}$, each \textit{non-guard} which is closest to $m$, moves towards $m$ either synchronously or there may be a possible \textit{pending move} due to the asynchronous behavior of the scheduler. Ties are broken by considering the closest robots which are farthest from the \textit{leading corners} in their respective \textit{string directions}.
	     \item $C(t) \in \mathcal I_4:$ First consider the case when the \textit{target Weber meeting node} $m$ is the center of rotational symmetry. Each robot moves towards $m$. 
	     
	     \noindent Next, consider the case when $C(t) \in \mathcal I_4^a$ and the \textit{target Weber meeting node} $m$ is on $\mathcal H^{++}$. After the \textit{Leading Robot Selection} phase, each \textit{non-guard} robot in the quadrants different from $\mathcal H^{++}$ as well as on the lines $l$ or $l'$ moves towards the \textit{target Weber meeting node} $m$ in $\mathcal H^{++}$. First, the \textit{non-guards} in the quadrants adjacent to $\mathcal H^{++}$ and the robots on $l$ or $l'$ moves towards $\mathcal H^{++}$. While they reach $\mathcal H^{++}$, the \textit{non-guards} in the quadrants non-adjacent to $\mathcal H^{++}$ moves towards $\mathcal H^{++}$. This movement is required in order to ensure that $\mathcal H^{++}$ remains invariant. While such \textit{non-guards} reach $\mathcal H^{++}$, each \textit{non-guards} in $\mathcal H^{++}$ moves towards $m$, thus creating a robot multiplicity on $m$. Otherwise, consider the case when $m$ is on either on $l$ or $l'$. First all the \textit{non-guards} in $\mathcal H^{++}$ moves towards $m$. While they reach $l$ or $l'$, the other \textit{non-guards} moves towards $m$. This movement is required in order to ensure that $\mathcal H^{++}$ remains invariant. Finally, a robot multiplicity is created at $m$.
	     
	     \noindent If there exist eight \textit{Potential Weber meeting nodes}, then there exist exactly two \textit{Potential Weber meeting nodes} in $\mathcal H^{++}$. There exist two \textit{string directions} corresponding to the unique \textit{leading corner} in $\mathcal H^{++}$. The robot closest to the \textit{target meeting node} $m$ and appearing first in $\alpha_i$, which is not on $m$, moves towards $m$. After such a move of the robot, it ensures that $m$ remains invariant during the procedure. The procedure proceeds similar to as before.
	    \end{enumerate}
	    Note that if the configuration is asymmetric or symmetric with respect to rotational symmetry, there may exist at most four robot positions that are not on $m$ during this phase. Otherwise, if the configuration is symmetric with respect to a horizontal or vertical line of symmetry, there may exist at most six robot positions that are not on $m$.
\subsubsection{Finalization of Gathering} 
Let $m$ be the \textit{target Weber meeting node}, where a robot multiplicity is created during the \textit{Creating Multiplicity on Target Weber meeting node} phase. The following are the cases in which the robots will identify that the {\it Finalization of Gathering} is in progress:
	\begin{enumerate}
	    \item The configuration has at most four robot positions that are not on $m$ containing robot multiplicity. Moreover, each side of the $M E R $ contains at most one robot position.
	    \item The configuration has exactly six robot positions that are not on $m$ containing robot multiplicity. Moreover, there exist exactly two sides of the $M E R $ that contain two robot positions and are symmetric images of each other.
	    \end{enumerate}

\noindent In this phase, all the \textit{guards} move towards $m$. During their movement, they do not create any multiplicity on a \textit{Weber meeting node} other than $m$. In order to ensure this, all the \textit{guards} first move along the boundary of $MER$, and when it becomes collinear with $m$, it starts moving towards $m$. A \textit{guard} robot moves by minimizing the Manhattan distance between $m$ and itself. This implies that during their movement, no other multiplicity would be created on any other \textit{Weber meeting node} and \textit{gathering} would be finalized on $m$.
\subsection{Optimal Gathering()}
Our main algorithm \textit{Optimal Gathering()} considers the following cases. 
If $C(t) \in \mathcal I_1$, then each robot finalizes the \textit{gathering} on the unique \textit{Weber meeting node}.

\noindent Consider the case when the \textit{meeting nodes} are asymmetric. There exists a unique \textit{Potential Weber meeting node}. The \textit{guards} are selected in the \textit{Guard Selection} phase. Each \textit{non-guard} moves towards the unique \textit{Potential Weber meeting node}, creating a multiplicity on it. Finally, the \textit{guards} moves towards the multiplicity and finalizes the \textit{gathering} on it.

\noindent Next, consider the case when the configuration is balanced and asymmetric. A \textit{leading robot} is selected in the \textit{Leading Robot Selection} phase, which transforms the configuration into an unbalanced configuration. The \textit{guards} are selected in the \textit{Guard Selection phase}. In the \textit{Creating Multiplicity on the Target Weber meeting node} phase, each \textit{non-guard} moves towards the \textit{target Weber meeting node}, selected in the \textit{Target Weber meeting node Selection} phase. Finally, the \textit{guards} moves towards the multiplicity and finalizes the \textit{gathering}.

\noindent If $C(t)$ is symmetric and there exists a \textit{Weber meeting node} on $l \cup \lbrace c \rbrace$, the \textit{gathering} is finalized on the \textit{target Weber meeting node} $m$ selected in the \textit{Target Weber meeting node Selection} phase. Otherwise, if the configuration is symmetric and there exists a robot on either $l$ or $c$, then in the \textit{Symmetry Breaking} phase, the configuration is transformed into an asymmetric configuration. Note that, in case $C(t) \in \mathcal I_3^b$, there may exist exactly six robots that are selected as \textit{guards}. In case $n=7$, there must exist at least one robot position on $l$. The northernmost robot on $l$ moves towards an adjacent node away from $l$, if there does not exist any \textit{Weber meeting nodes} on $l$. Hence, the configuration becomes asymmetric, and the algorithm proceeds similarly, as in the asymmetric case for $n=7$. Otherwise, if there exists at least one \textit{Weber meeting node} on $l$, the northernmost \textit{Weber meeting node} $m$ on $l$ is selected as the \textit{target Weber meeting node}. The closest robot on $l$ and the northernmost in case of a tie, moves towards $m$. While the robot moves towards $m$, it remains invariant. After the robot reaches $m$, $m$ is uniquely identifiable and the \textit{gathering} is finalized in the \textit{Finalization of Gathering} phase.
\section{Correctness} \label{s5}
In this section, we describe the correctness of our proposed algorithm. Lemmas \ref{correc1} and \ref{correc2} proves that the \textit{leading robot} remains invariant during the movement towards its destination.  
	\begin{lemma}
		If $C(t) \in \mathcal I_3^a$, then in the \textit{Leading Robot Selection} phase, the \textit{leading robot} remains the unique robot while it moves towards the half-plane $\mathcal H^{+}$. 
		\label{correc1}
	\end{lemma}
	\begin{proof}
		Let $C(t)$ be any balanced configuration that belongs to $\mathcal I_3^{a}$. Since the configuration is balanced and asymmetric, the number of robots in the two half-planes delimited by $l$ are equal and there exists a unique \textit{key corner}. If there exists at least one robot position on $l$, then the northernmost robot on $l$ is the \textit{leading robot}. The northernmost robot moves towards an adjacent node away from $l$, and the configuration becomes unbalanced. Consider the case when there does not exist any robot position on $l$. Without loss of generality, assume that $l$ is a vertical line of symmetry. Let $r$ be the \textit{leading robot} in $\mathcal H^{-}$ selected in the \textit{Leading Robot Selection} phase. Without loss of generality, let $A$ be the unique \textit{key corner} and $\alpha_{AD}=a_1, a_2,\ldots, a_{p q}$ is the unique smallest lexicographic string associated to the corner $A$. Similarly, let $B$ be the other \textit{leading corner} and $\alpha_{BC}=b_1, b_2,\ldots, b_{p q}$ be the string associated to $B$. Let $u_i$ and $v_i$ denote the nodes, which the positions $a_i$ and $b_i$ represent in $\alpha_{AD}$ and $\alpha_{BC}$, respectively. Since the \textit{meeting nodes} are symmetric, $f_t(u_i)=f_t(v_i)$, for each $i= 1, 2 \ldots, p q$. As $\alpha_{AD}=a_1, a_2,\ldots, a_{p q}$ is the unique smallest lexicographic string among the $\alpha_i's$, there must exist a position $k'$ such that $\lambda_t(u_{k'})=0< \lambda_t(v_{k'})=1$. Without loss of generality, let $i$ be the position of the \textit{leading robot} in $\alpha_{AD}$. Let $k$ be the first position, where $\lambda_t(u_{k})$ and $\lambda_t(v_{k})$ differs. Note that, $\lambda_t(u_{k})$=0 and $\lambda_t(v_{k})=1$. We have to prove that after the movement of the \textit{leading robot}, $\alpha_{AD} <_l \alpha_{BC}$, where $'<_l'$ denotes the relation that $\alpha_{AD}$ is lexicographically smaller than $\alpha_{BC}$. Assume that at time $t'$, the \textit{leading robot} moves towards an adjacent node. Depending on the possible values of $i$ and $k$ in $\alpha_{AD}$, the following cases are considered. 
		\setcounter{case}{0}
		\begin{case}
		\normalfont
		The position of $i$ is less than $k$ in $\alpha_{AD}$. While the \textit{leading robot} moves towards $l$, $\lambda_t(u_i)$ becomes 0, but $\lambda_t(v_i)$ equals 1. Hence, after the movement of the \textit{leading robot} towards an adjacent node, $\alpha_{AD} <_l \alpha_{BC}$.  
		\end{case}
		\begin{case}
		\normalfont
		The position of $i$ is equal to $k$ in $\alpha_{AD}$. Since each robot is deployed at the distinct nodes of the grid in the initial configuration, this case is not possible.
		\end{case}
		\begin{case}
		\normalfont
		The position of $i$ is greater than $k$ in $\alpha_{AD}$. While the \textit{leading robot} moves towards $l$, the position $k$ remains invariant. Hence, after the movement of the \textit{leading robot} towards an adjacent node, $\alpha_{AD} <_l \alpha_{BC}$.  
		\end{case}
		Note that, after a single movement of the \textit{leading robot} towards $l$, it becomes the unique robot that is eligible to move towards $\mathcal H^{+}$. Since $\alpha_{AD}$ remains the unique lexicographically smallest string at $t'$, $\mathcal H^{+}$ remains invariant. Clearly, after a finite number of movements towards $l$, $H^{+}$ remains invariant, and ultimately, the configuration becomes unbalanced. The proof is similar when the \textit{meeting nodes} admits a horizontal or a diagonal line of symmetry.
	\end{proof}
		\begin{lemma}
		If $C(t) \in \mathcal I_4^a $, then in the \textit{Leading Robot Selection} phase, the \textit{leading robot} remains the unique robot while it moves towards the \textit{target Weber meeting node}. \label{correc2}
	\end{lemma}
	\begin{proof}
	Let $C(t)$ be any balanced configuration that belongs to $\mathcal I_4^{a}$. Since the configuration is balanced and asymmetric, there exist at least two quadrants that contain the maximum number of \textit{Potential Weber meeting nodes} with the maximum number of robots on such quadrants. We have to prove that while the \textit{leading robot} moves towards the \textit{target Weber meeting node}, the quadrant $\mathcal H^{++}$ remains invariant. First, consider the case when the \textit{leading robot} $r$ is on either $l$ or $l'$. Note that in this case, $r$ may be one or more than one node away from $\mathcal H^{++}$. There is nothing to prove when $r$ is one node away from $\mathcal H^{++}$. In this case, a move of $r$ transforms the configuration into an unbalanced configuration. Therefore, consider the case when $r$ is more than one node away from $\mathcal H^{++}$. Without loss of generality, let $r$ be on $l$. Let $MER=ABCD$ be such that the corner $C$ is the corner diagonally opposite to $A$ and the corners $A$ and $B$ are separated by line $l$. Similarly, $A$ and $D$ are the corners separated by line $l'$. $\mathcal H^{++}$ is the quadrant containing $A$. Let $\alpha_{AD}=a_1, a_2, \ldots a_{pq}$ and $\alpha_{BC}=b_1, b_2, \ldots b_{pq}$ be the strings associated to the corners $A$ and $B$. While $r$ moves along $l$, we have to prove that that $\alpha_{AD}$ remains lexicographic larger than $\alpha_{BC}$. It is noteworthy that $\alpha_{AD}$ is lexicographic larger than $\alpha_{CB}$ and $\alpha_{AD}$ while $r$ moves. Note that, we have consider the case when the \textit{string directions} are along the width of the rectangle. Let $i$ be the position of \textit{leading robot} in $\alpha_{AD}$ and $\alpha_{BC}$. Let $u_i$ and $v_i$ denote the nodes, which the positions $a_i$ and $b_i$ represent in $\alpha_{AD}$ and $\alpha_{BC}$, respectively. Since the \textit{meeting nodes} are symmetric, $f_t(u_i)=f_t(v_i)$, for each $i= 1, 2 \ldots, p q$. After a movement of the \textit{leading robot} along the line $l$, note that $\mathcal H^{++}$ remains invariant. After a finite number of movements, the robot $r$ becomes one node away from $\mathcal H^{++}$, and the proof proceeds similarly as before. Next, consider the case when the \textit{leading} robot is on a quadrant adjacent to $\mathcal H^{++}$. Without loss of generality, assume that the \textit{leading robot} is on $\mathcal H^{+-}$. While the \textit{leading robot} moves, it can be observed that $\alpha_{AD}$ is lexicographically larger than $\alpha_{CB}$ and $\alpha_{DA}$. We have to prove that $\alpha_{AD}$ remains lexicographic larger than $\alpha_{BC}$ while the \textit{leading robot} moves. Let $i$ and $j$ be the positions of the \textit{leading corner} in $\alpha_{AD}$ and $\alpha_{BC}$, respectively. Note that $ i<j$, as the \textit{leading robot} is selected on $\mathcal H^{+-}$. Let $k$ be the first position for which $b_k < a_k$. We have the following cases.
		\setcounter{case}{0}
		\begin{case}
		\normalfont
		$i< j <k$. Note that $u_{i-1}$ cannot be a robot position, otherwise $r$ would not be selected as a \textit{leading robot}. After a move of $r$, $u_{i-1}$ is a robot position but $v_{i-1}$ cannot be a robot position.
		\end{case}
		\begin{case}
		\normalfont
		$i<j=k$. Since each robot is deployed at the distinct nodes of the grid in the initial configuration, this case is not possible.
		\end{case}
			\begin{case}
		\normalfont
		$i<k<j$. After a move of $r$, $u_{i-1}$ is a robot position, but $v_{i-1}$ cannot be a robot position, as $k$ is the first position where $a_k >b_k$.
		\end{case}
			\begin{case}
		\normalfont
		$i=k<j$. The proof is similar to the previous case.
		\end{case}
			\begin{case}
		\normalfont
		$k<i<j$. After a move of $r$, it may be the case that $k=i-1$ in $\alpha_{AD}$. In that case, $\lambda_{t}(u_{i-1})=2$, but $\lambda_{t}(v_{i-1})=0$. Otherwise, the proof is similar as $a_k>b_k$.
		\end{case}
		The proof is similar when the \textit{string directions} are along the lengths of $MER$. Next, consider the case when the \textit{leading robot} $r$ is selected on a quadrant non-adjacent to $\mathcal H^{++}$. Without loss of generality, we assume that $r$ first starts moving towards $l'$. While the \textit{leading robot} moves, it can be observed that $\alpha_{AD}$ is lexicographically larger than $\alpha_{CB}$ and $\alpha_{DA}$. We have to prove that $\alpha_{AD}$ remains lexicographic larger than $\alpha_{BC}$ while $r$ moves. Let $i$ and $j$ be the positions of the \textit{leading corner} in $\alpha_{AD}$ and $\alpha_{BC}$, respectively. Note that $ i>j$, as the \textit{leading robot} is selected on $\mathcal H^{--}$. 
		We have the following cases.
		\setcounter{case}{0}
		\begin{case}
		\normalfont
		$i> j >k$. Note that $u_{i-1}$ cannot be a robot position, otherwise $r$ would not be selected as a \textit{leading robot}. After a move of $r$, $a_{i-1} \geq b_{i-1}$, depending on whether there exists a robot position on $b_{i-1}$ or not.
		\end{case}
		\begin{case}
		\normalfont
		$i>j=k$. Since each robot is deployed at the distinct nodes of the grid in the initial configuration, this case is not possible.
		\end{case}
			\begin{case}
		\normalfont
		$i=k>j$. After a move of $r$, $a_{i-1} \geq b_{i-1}$, depending on whether there exists a robot position on $b_{i-1}$ or not.
		\end{case}
			\begin{case}
		\normalfont
		$k>i>j$. Note that $a_{i-1}$ cannot be a robot position before the move, otherwise $r$ would not be selected as a robot position. After a move of $r$, $a_{i-1}$ is a robot position, but $b_{i-1}$ cannot be a robot position as $k$ is the first position where $a_k$ and $b_k$ differ. 
		\end{case}
			\begin{case}
		\normalfont
		$j<k<i$. After a move of $r$, it may be the case that $k=i-1$ in $\alpha_{AD}$. In that case, $\lambda_{t}(u_{i-1})=2$, but $\lambda_{t}(v_{i-1})=0$. Otherwise, the proof is similar as $a_k>b_k$.
		\end{case}
		From all the above cases, the \textit{leading robot} remains invariant while it moves towards its destination.
		\end{proof}
		The next three lemmas prove that the \textit{target Weber meeting node} remains invariant in the \textit{Creating Multiplicity on the Target Weber meeting node} phase.
		\begin{lemma}
		If $C(t)\in$ $\mathcal I_2\cup \mathcal I_3^{b1} \cup \mathcal I_4^{b1} $, then the \textit{target Weber meeting node} remains invariant in the \textit{Creating Multiplicity on Target Weber meeting node} phase.\label{correc3}
	\end{lemma}
	\begin{proof}
 	In the \textit{Creating Multiplicity on Target Weber meeting node} phase, all the \textit{non-guard} robots move towards the \textit{target Weber meeting node}. According to Lemma \ref{m1}, the \textit{Weber meeting node} remains invariant under the movement of robots towards itself. The $MER$ remains invariant unless the \textit{guard} moves. The following cases are to be considered.
	\setcounter{case}{0}
	\begin{case}
	\normalfont
	$C(t) \in \mathcal I_2$: Since the \textit{meeting nodes} are asymmetric, there exists a unique \textit{Potential Weber meeting node}. The unique \textit{Potential Weber meeting node} is selected as the \textit{target Weber meeting node}. The unique \textit{Potential Weber meeting node} of the configuration is defined with respect to the position of the \textit{leading corner}. The \textit{leading corner} remains invariant unless the $MER$ changes. As the \textit{guards} does not move in the \textit{Creating Multiplicity on Target Weber meeting node} phase, the $MER$ remains invariant. Hence, the \textit{target Weber meeting node} remains invariant.  
	\end{case}	
\begin{case}
		 \normalfont
		 $C(t) \in \mathcal I_3^{b1}$: The northernmost \textit{Weber meeting node} on $l$ is selected as the \textit{target Weber meeting node}. Since the northernmost agreement depends on the position of the \textit{leading corner(s)}, the agreement remains invariant unless the $MER$ changes. As the \textit{guards} does not move in the \textit{Creating Multiplicity on Target Weber meeting node} phase, the $MER$ remains invariant. Hence, \textit{target Weber meeting node} remains invariant. 
		 \end{case} 
		 \begin{case} 
		 \normalfont
		 $C(t)\in \mathcal I_4^{b1}$: The center of rotational symmetry $c$ is the \textit{target Weber meeting node}. Since $c$ is also the center of rotational symmetry for the \textit{meeting nodes} also, the \textit{target Weber meeting node} remains invariant.
		 \end{case}
	\end{proof}
		\begin{lemma}
		If $C(t)\in$ $\mathcal I_3^a$, then the \textit{target Weber meeting node} remains invariant in the \textit{Creating Multiplicity on Target Weber meeting node} phase.\label{correc4}
	\end{lemma}
	\begin{proof}
		The \textit{meeting nodes} are symmetric with respect to a single line of symmetry $l$. In the \textit{Creating Multiplicity on Target Weber meeting node} phase, all the \textit{non-guard} robots move towards the \textit{target Weber meeting node}. The $MER$ remains invariant unless the \textit{guard} moves. The following cases are to be considered.
		\setcounter{case}{0}
	\begin{case}
	\normalfont
		 There exists at least one \textit{Weber meeting node} on $l$. The northernmost \textit{Weber meeting node} on $l$ is selected as the \textit{target Weber meeting node}. Since the northernmost agreement depends on the position of the \textit{leading corner(s)}, and the \textit{leading corner(s)} remains invariant unless the $MER$ changes, the agreement remains invariant. Hence, the \textit{target Weber meeting node} remains invariant.
		 \end{case}
	\begin{case}
	\normalfont
		$C(t)$ satisfy $C_1$. In this case, we have to prove that  $\mathcal H^+$ remains invariant in the \textit{Creating Multiplicity on Target Weber meeting node} phase. Note that, in this phase, all the \textit{non-guards} move towards the \textit{target Weber meeting node}. According to Lemma \ref{wbn}, the \textit{Weber meeting nodes} remains invariant while the robots move towards it. As $M E R $ remains invariant unless the \textit{guard} robot moves, the \textit{leading corner(s)} remains invariant. Since the \textit{Potential Weber meeting nodes} are defined with respect to the positions of the \textit{leading corner(s)}, $\mathcal H^+$ remains invariant. Hence, the \textit{target Weber meeting node} remains invariant.
		\end{case}
		\begin{case}
		\normalfont
		$C(t)$ satisfy $C_2$. The following subcases are to be considered. 
		 \setcounter{subcase}{0}
		\begin{subcase}
		\normalfont
			$C(t)$ satisfy $C_{22}$. $\mathcal H^{+}$ is the half-plane that contains the maximum number of robots. All the \textit{non-guard} robots in $\mathcal H^{-}$ move towards the \textit{target Weber meeting node} in $\mathcal H^{+}$. During this movement of the robots, $\mathcal H^{+}$ still contains the maximum number of robots. Hence, the \textit{target Weber meeting node} remains invariant.
			\end{subcase}
			\begin{subcase}
			\normalfont
			 $C(t)$ satisfy $C_{21}$. The \textit{leading robot} in $\mathcal H^{-}$ moves towards the \textit{target Weber meeting node} in $\mathcal H^{+}$ resulting in transforming the configuration into an unbalanced configuration. While the \textit{leading robot} moves towards $\mathcal H^{+}$, the unique lexicographic smallest string $\alpha_i$ remains invariant according to Lemma \ref{correc1}. As the \textit{key corner} remains invariant, $\mathcal H^{+}$ remains invariant. The moment the \textit{leading robot} reaches $l$, the configuration becomes unbalanced. The rest of the proof follows from the previous case.
		\end{subcase}
		
	\end{case}
	\end{proof}
		\begin{lemma}
		If $C(t)\in$ $\mathcal I_4^a$, then the \textit{target Weber meeting node} remains invariant in the \textit{Creating Multiplicity on Target Weber meeting node} phase.\label{correc5}
	\end{lemma}
	\begin{proof}
		Since $C(t)\in$ $\mathcal I_4^{a}$, the \textit{meeting nodes} are symmetric with respect to rotational symmetry. Let $c$ be the center of the rotational symmetry for $M$. According to Lemma \ref{wbn}, the \textit{Weber meeting nodes} remains invariant while all the robots move towards it. The following cases are to be considered.
		\setcounter{case}{0}
	\begin{case}
	    \normalfont
	There exists a \textit{Weber meeting node} on $c$. It is selected as the \textit{target Weber meeting node}. Since $c$ is the center of rotational symmetry for the fixed \textit{meeting nodes}, the \textit{target Weber meeting node} remains invariant while the robots move towards it.
	\end{case}
	\begin{case}
	\normalfont
	$C(0)$ satisfy $C_1$. The \textit{target Weber meeting node} is selected in $\mathcal H^{++}$ as the \textit{Weber meeting node} which is farthest from the \textit{leading corner} contained in $\mathcal H^{++}$ in the \textit{string direction}. We have to prove that $\mathcal H^{++}$ remains invariant while the robots move towards the \textit{target Weber meeting node}. Note that since the \textit{guards} do not move during this phase, the $MER$ remains invariant. As a result, the \textit{Potential Weber meeting nodes} and $\mathcal H^{++}$ remain invariant. Hence, the \textit{target Weber meeting node} remains invariant.
	\end{case}
	\begin{case}
	\normalfont
	   $C(0)$ satisfy $C_2$. We have to prove that $\mathcal H^{++}$ remains invariant while the robots move towards the \textit{target Weber meeting node}. Considering such quadrants that contain the maximum number of \textit{Potential Weber meeting nodes}, the \textit{target Weber meeting node} is selected as the \textit{Potential Weber meeting node} in $\mathcal H^{++}$. Ties are broken by considering the \textit{Weber meeting node} in $\mathcal H^{++}$ which is farthest from the \textit{leading corner} in $\mathcal H^{++}$ in the \textit{string direction}. If $C(0)$ satisfies $C_{22}$, first, all the \textit{non-guard} robots in the quadrants adjacent to $\mathcal H^{++}$ and on $l \cup l'$ move towards the \textit{target Weber meeting node} in $\mathcal H^{++}$. Finally, the other \textit{non-guard} robots move towards $m$. Since $\mathcal H^{++}$ is the unique quadrant that contains the maximum number of robot positions, it still contains the maximum number of robots while all such robots reach $\mathcal H^{++}$. Hence, the \textit{target Weber meeting node} remains invariant. Otherwise, if $C(0)$ satisfies $C_{21}$, there exists more than one quadrant that contains the maximum number of robot positions. Considering such quadrants and the corners contained in those quadrants. $\mathcal H^{++}$ is the quadrant containing the largest lexicographic string among those $\alpha _i' s$ that are associated with the \textit{leading corners} contained in such quadrants. A \textit{leading robot} is selected in the \textit{Leading Robot Selection} phase and is allowed to move towards the \textit{target Weber meeting node} in $\mathcal H^{++}$. While the \textit{leading robots} moves towards the \textit{target Weber meeting node} in $\mathcal H^{++}$, $\mathcal H^{++}$ remains invariant according to Lemma \ref{correc2}. As a result, the configuration becomes unbalanced. The rest of the proof follows similarly, as in the unbalanced case.      
	\end{case}
	\begin{case}
\normalfont			
 $C(0)$ satisfy $C_3$. The \textit{target Weber meeting node} is selected on either $l$ or $l'$. Note that, in this case, if the configuration is unbalanced, $\mathcal H^{++}$ is the quadrant that contains the minimum number of robots. Otherwise, if the configuration is balanced, then $\mathcal H^{++}$ is the quadrant containing the smallest lexicographic string among all those $\alpha _i' s$. In both cases, all the \textit{non-guard} robots on $l \cup l'$, and the \textit{non-guard robots} on $\mathcal H^{++}$, move towards the \textit{target Weber meeting node} $m$. After such robots reach $\mathcal H^{++}$, $\mathcal H^{++}$ remains the unique quadrant with the minimum number of robots. As a result, $\mathcal H^{++}$ remains invariant. Hence, the \textit{target Weber meeting node} remains invariant. 
\end{case} 
\end{proof}
The next two lemmas prove that any initial configuration $C(0) \in \mathcal I \setminus \mathcal U$, would never reach a configuration $C(t) \in \mathcal U$, at any point of time $t>0$ during the execution of the algorithm \textit{Optimal Gathering()}.
	\begin{lemma}
		Given $C(0)\in \mathcal I_3$ and $t > 0$ be an arbitrary instant of time at which at least one robot has completed its LCM cycle. If $C(0)\notin \mathcal I_3^{b3}$, then during the execution of the algorithm \textit{Optimal Gathering()}, $C(t)\notin \mathcal I_3^{b3}$.\label{correc8} 
	\end{lemma}
	\begin{proof}
	According to Lemma \ref{m1}, the \textit{Weber meeting nodes} remains invariant while the robots move towards it. Since the \textit{meeting nodes} admits a single line of symmetry $l$ and there does not exist any \textit{Weber meeting node} on $l$, assume that $C(0) \in \mathcal I_3^a \cup \mathcal I_3 ^{b2}$. The following cases are to be considered.
	 \setcounter{case}{0}
	\begin{case}
	\normalfont
		$C(0) \in \mathcal I_3^a$. Note that there does not exist any \textit{Weber meeting node} on $l$, otherwise according to Lemma \ref{m1}, $C(t)\notin \mathcal I_3^{b3}$. Depending on the position of \textit{Potential Weber meeting nodes}, the following subcases may arise.
		 \setcounter{subcase}{0}
		\begin{subcase}
		\normalfont
			 $C(0)$ satisfy $C_1$. All the \textit{non-guard} robots in $\mathcal H^{-} \cup l$ move towards the \textit{target Weber meeting node} in $\mathcal H^{+}$. So, at any arbitrary instant of time $t>0$, $C(t)$ remains asymmetric and hence $C(t)\notin \mathcal I_3^{b3}$.
			 \end{subcase}
			 \begin{subcase}
			 \normalfont
			 $C(0)$ satisfy $C_2$. If the configuration is unbalanced, all the robots in $\mathcal H^{-} \cup l$ moves towards the \textit{target Weber meeting node} in $\mathcal H^{+}$. As a result, the configuration remains unbalanced and hence asymmetric. Otherwise, if the configuration is balanced, the \textit{leading robot} moves towards the \textit{target Weber meeting node} at some time $t'>0$. According to Lemma \ref{correc1}, the configuration remains asymmetric during its movement and ultimately, the configuration becomes unbalanced. Proceeding similarly, as in the unbalanced case, at any arbitrary instant of time $t>0$, $C(t)$ remains asymmetric and hence $C(t)\notin \mathcal I_3^{b3}$, where $t \geq t'$.
		\end{subcase}
		\end{case}
		\begin{case}
		\normalfont
		 $C(0) \in \mathcal I_3 ^{b2}$. Assume that at time $t'>0$, the northernmost robot on $l$ moves towards an adjacent node away from $l$, which transforms the configuration into an unbalanced asymmetric configuration. The rest of the proof follows from the previous case. Hence, $C(t)\notin \mathcal I_3^{b3}$, where $t'\geq t$.
	\end{case}
	\end{proof}
	\begin{lemma}
		Given $C(0)\in \mathcal I_4$ and $t > 0$ be an arbitrary instant of time at which at least one robot has completed its LCM cycle. If $C(0)\notin \mathcal I_4^{b3}$, then during the execution of the algorithm Optimal Gathering(), $C(t)\notin \mathcal I_4^{b3} $. \label{correc9}
	\end{lemma}
	\begin{proof}
		According to Lemma \ref{m1}, the \textit{Weber meeting nodes} remains invariant while the robots move towards it. Since the \textit{meeting nodes} admit rotational symmetry and there does not exist any \textit{Weber meeting node} on $c$, assume that $C(0) \in \mathcal I_4 ^a \cup \mathcal I_4 ^{b2}$. The following cases are to be considered.
		 \setcounter{case}{0}
	\begin{case}
	\normalfont
	    $C(0) \in \mathcal I_4 ^a$. If there exists a \textit{Weber meeting node} on $c$, then all the robots move towards it and finalize the \textit{gathering}. According to Lemma \ref{wbn}, since the \textit{Weber meeting node} remains invariant while all the robots move towards it, $C(t)\notin \mathcal I_4^{b3}$. Consider the case when there does not exist any \textit{Weber meeting node} on $c$. The following subcases may arise.
	     \setcounter{subcase}{0}
	    \begin{subcase}
	    \normalfont
	    	$C(0)$ satisfy $C_1$. All the \textit{non-guard} robots from the other quadrants as well on $l$ or $l'$ move towards the \textit{target Weber meeting node} in $\mathcal H^{++}$. So, at any arbitrary instant of time $t>0$, $C(t)$ remains asymmetric and hence $C(t)\notin \mathcal I_4^{b3}$.
	    	\end{subcase}
	    	\begin{subcase}
	    	\normalfont
	    	    $C(0)$ satisfy $C_2$. If the configuration is unbalanced, all the robots in the quadrants different from $\mathcal H^{++}$ as well as the robots on $l$ or $l'$ move towards the \textit{target Weber meeting node} in $\mathcal H^{++}$ in the \textit{Creating Multiplicity on Target Weber meeting node} phase. While such a robot reaches $\mathcal H^{++}$, the configuration remains unbalanced and hence asymmetric. If the configuration is balanced, a \textit{leading robot} is selected in the \textit{Leading Robot Selection} phase. According to Lemma \ref{correc2}, the configuration remains asymmetric during the movement of the \textit{leading robot} towards the \textit{target Weber meeting node} at some time $t'>0$. While the \textit{leading robot} reaches $\mathcal H^{++}$, the configuration becomes unbalanced and remains asymmetric. So, at any arbitrary instant of time $t>0$, $C(t)\notin \mathcal I_4^{b3}$, where $t\geq t'$.
	    	\end{subcase}

\begin{subcase}
\normalfont
$C(0)$ satisfy $C_3$. All the robots in $\mathcal H^{++}$, move towards the \textit{target Weber meeting node}. After all the robots in $\mathcal H^{++}$, reach the \textit{target Weber meeting node} $m$, all the \textit{non-guard} robots from the other quadrants as well as on $l$ or $l'$ move towards $m$, thus creating a multiplicity on $m$. During this robot movement, $C(t)$ remains asymmetric and hence $C(t)\notin \mathcal I_4^{b3}$. 
\end{subcase}		 
	\end{case}
	\begin{case}
	\normalfont
	$C(0) \in \mathcal I_4 ^{b2}$. Assume that at time $t'>0$, the robot on $c$ move towards one of the adjacent nodes which transforms the configuration into a configuration which may be asymmetric or admits a single line of symmetry. Proceeding similarly, as in the case of $C(0) \in \mathcal I_3^{a} \cup \mathcal I_4 ^{a}$, at any arbitrary instant of time $t>0$, $C(t)$ remains asymmetric and hence $C(t)\notin \mathcal I_4^{b3}$, where $t \geq t'$.
	\end{case}
	\end{proof}

\begin{theorem}
If the initial configuration belongs to the set $\mathcal I \setminus U$, then algorithm \textit{Optimal Gathering()} ensures \textit{gathering over Weber meeting nodes}.
\end{theorem}
\begin{proof}
Assume that $C(0) \in \mathcal I \setminus \mathcal U$. If $C(t)$ is not a final configuration for some $t\geq 0$, each active robot executes algorithm \textit{Optimal Gathering()}. According to the Lemmas \ref{correc8} and \ref{correc9}, any initial configuration $C(0) \in \mathcal I \setminus \mathcal U$, would never reach a configuration $C(t) \in \mathcal U$, at any point of time $t>0$ during the execution of the algorithm \textit{Optimal Gathering()}. The following cases are to be considered.
 \setcounter{case}{0}
\begin{case}
\normalfont
There exists a unique \textit{Weber meeting node}. All the robots move towards the unique \textit{Weber meeting node} and finalize the \textit{gathering}.
\end{case}
\begin{case}
\normalfont
There exists more than one \textit{Weber meeting node}. The \textit{target Weber meeting node} is selected in \textit{Target Weber meeting node Selection} phase. According to the Lemmas \ref{correc3}, \ref{correc4} and \ref{correc5}, the \textit{target Weber meeting node} remains invariant during the execution of the algorithm \textit{Optimal Gathering()}. If $C(0)$ is a balanced configuration, then a \textit{leading robot} is selected in \textit{Leading Robot Selection} phase. Lemmas \ref{correc1} and \ref{correc2} ensure that the \textit{leading robot} remains invariant during its movement. 

\noindent Without loss of generality, assume that $m$ is the \textit{target Weber meeting node}. Assume that, at any point of time $t$, there exists at which at least one robot $r$ that has completed its LCM cycle. If $r$ is a \textit{non-guard} robot, then it must have moved at least one unit distance towards $m$ at time $t'>t$. Since, each \textit{non-guard} robot moves towards $m$ via a shortest path in the \textit{Creating Multiplicity on Target Weber meeting node} phase, this implies that eventually at time $t''>t'$, there exists a robot multiplicity on $m$. Finally, in the \textit{Finalization of Gathering} phase, since the robots have \textit{global strong-multiplicity detection} capability, all the \textit{guard} robots move towards $m$ and finalize the \textit{gathering} without creating any other multiplicity on a \textit{meeting node}. Since each robot finalizes the \textit{gathering}, by moving towards $m$ via a shortest path, \textit{gathering over Weber meeting nodes} is ensured. 
\end{case}

\end{proof}	
	\begin{figure}[h]
	
			\centering
		
			{
				\includegraphics[width=0.750\columnwidth]{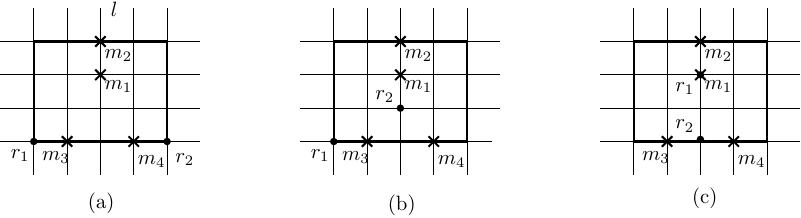}
			}
			\caption{(a) $C(0)$, (b) $C(t_1)$, (c) $C(t_2)$} 
	        \label{rep}
	        \end{figure}

\section{Optimal Gathering for $C(t)\in \mathcal{U}$} \label{s6}
We have proposed a deterministic distributed algorithm that ensures {\it gathering} over a \textit{Weber meeting node} for any initial configuration $C(0)\in \mathcal{I}\setminus \mathcal{U}$. Let $\mathcal{U}'\subset \mathcal{U}$ denote the set of all the initial configurations which admit a unique line of symmetry $l$ and no \textit{Weber meeting nodes} or robot positions exist on $l$. However, there exists at least one \textit{meeting node} on $l$. The set $\mathcal{U}'$ includes the initial configurations for which {\it gathering} is feasible on a \textit{meeting node}. Note that, if $C(t)\in \mathcal{U}\setminus\mathcal{U}'$, then it is ungatherable. To ensure \textit{gathering} deterministically, the target point must lie on $l$. At this point of time, one optimal feasible solution for a configuration $C(0)\in \mathcal{U}'$ would be to finalize the \textit{gathering} at a \textit{meeting node} $m\in l$ at which the total number of moves is minimized. Ties may be broken by considering the northernmost such \textit{meeting node}. Another very important assumption that is not highlighted much in the literature is that initially, all the robots are static. The correctness of our proposed algorithm fails to hold when the optimal target point is dynamically selected. As a consequence, termination may not be guaranteed with optimal number of moves. For example, we consider one possible execution for an initial configuration $C(0)=(\lbrace r_1,r_2\rbrace,\lbrace m_1,m_2,m_3,m_4\rbrace)$ in figure \ref{rep}(a). At $t=0$, $m_3$ and $m_4$ are the \textit{Weber meeting nodes}. Between $m_1$ and $m_2$, the number of total moves will be minimized if the robots gather at $m_1$. While $r_1$ and $r_2$ start moving towards $m_1$, there may be a pending move due to the asynchronous behavior of the scheduler. Consider the case when $r_2$ has completed its LCM cycle while $r_1$'s move is pending. At $t=t_1>0$, $m_3$ becomes the unique \textit{Weber meeting node} (figure~\ref{rep}(b)). At $t_2>t_1$, assume that $r_1$ has reached $m_1$ and $r_2$ has moved by one hop distance towards $m_3$. At $t_2$, $m_1$ becomes the unique Weber meeting node (figure~\ref{rep}(c)). Next, the \textit{gathering} will be finalized eventually at $m_1$. Initially, the minimum number of moves required to finalize the \textit{gathering} is 8 (figure~\ref{rep}(a)). The number of moves required to finalize the \textit{gathering} in this execution is 10. It is not guaranteed that the minimum number of moves required to finalize the \textit{gathering} in the initial configuration is achievable. 
\section{Conclusion} \label{s7}
In this paper, the \textit{optimal gathering over Weber meeting nodes} problem has been investigated over an infinite grid. The objective function is to minimize the total distance traveled by all the robots. We have characterized all the configurations for which \textit{gathering} over a \textit{Weber meeting node} cannot be ensured. For the remaining configurations, a deterministic distributed algorithm has been proposed that solves the \textit{gathering} over \textit{Weber meeting nodes} for at least seven robots.
	
\noindent One future direction of work would be to consider the \textit{min-max gathering over meeting nodes} problem, where the objective function is to minimize the maximum distance traveled by a robot. Since there remain some initial symmetric configurations, for which \textit{gathering over} \textit{Weber meeting nodes} cannot be ensured, it would be interesting to consider randomized algorithms for those configurations. Another direction for future interest would be to consider multiplicities in the initial configuration. 
 \bibliographystyle{ws-ijfcs}
\bibliography{elsarticle-template}
\end{document}